\documentclass[12pt]{article}

\usepackage[utf8]{inputenc}
\usepackage[T1]{fontenc}
\usepackage{times}
\usepackage[backref,colorlinks,citecolor=blue,bookmarks=true]{hyperref}
\usepackage{microtype}
\usepackage{amsmath}
\usepackage{amssymb}
\usepackage{amsthm}
\usepackage{bm}
\usepackage{bbm}
\usepackage{dsfont}
\usepackage[affil-it]{authblk}
\usepackage{fullpage}
\usepackage{comment}
\usepackage{tikz}
\usepackage{mathtools}
\usepackage[shortlabels]{enumitem}
\usepackage{complexity}
\usepackage[capitalize]{cleveref}
\crefname{ineq}{inequality}{inequalities}
\creflabelformat{ineq}{#2{\upshape(#1)}#3}
\usepackage{nag}
\usepackage{xspace}

\theoremstyle{plain}
\newtheorem{theorem}{Theorem}[section]
\newtheorem{lemma}[theorem]{Lemma}
\newtheorem{corollary}[theorem]{Corollary}

\theoremstyle{definition}
\newtheorem{definition}[theorem]{Definition}

\theoremstyle{remark}

\newcommand*{\N}{{\mathbb{N}}}

\let\R\relax
\newcommand*{\R}{{\mathbb{R}}}

\let\C\relax
\newcommand*{\C}{{\mathbb{C}}}
\newcommand*{\bigsum}{\sum\limits}
\newcommand*{\bigprod}{\prod\limits}
\newcommand*\diff{\mathop{}\!\mathrm{d}}
\newcommand*{\lf}{\mathop{\lfloor\!}}
\newcommand*{\rf}{\mathop{\!\rfloor}}
\newcommand*{\Per}{\operatorname{Per}}
\let\Pr\relax
\newcommand*{\Pr}{\mathbb{P}}
\let\Ex\relax
\newcommand*{\Ex}{\mathbb{E}}

\let\class\relax
\newcommand{\class}[1]{\ensuremath{\mathsf{#1}}\xspace}

\let\poly\relax
\DeclareMathOperator{\poly}{poly}
\let\polylog\relax
\DeclareMathOperator{\polylog}{polylog}
\let\polyloglog\relax
\DeclareMathOperator{\polyloglog}{polyloglog}
\let\Im\relax
\DeclareMathOperator{\Im}{Im}
\let\Re\relax
\DeclareMathOperator{\Re}{Re}
\DeclareMathOperator{\Var}{Var}

\newcommand{\Matrix}{\mathrm{M}}
\newcommand{\eps}{\varepsilon}
\newcommand{\abs}[1]{\left\vert {#1} \right\vert}
\newcommand{\iid}{i.i.d.\xspace}
\newcommand{\cut}{t}

\title{Approximating Permanent of Random Matrices with Vanishing Mean: Made
  Better and Simpler}

\author[1]{Zhengfeng Ji\footnote{Email: zhengfeng.ji@uts.edu.au}}
\author[2]{Zhihan Jin\footnote{Email: pascalprimer@sjtu.edu.cn}}
\author[3]{Pinyan Lu\footnote{Email: lu.pinyan@mail.shufe.edu.cn}}
\affil[1]{Centre for Quantum Software and Information, University of Technology Sydney}
\affil[2]{Zhiyuan College, Shanghai Jiao Tong University}
\affil[3]{Institute for Theoretical Computer Science, Shanghai University of Finance and Economics}

\date{}

\begin{document}

\maketitle

\begin{abstract}
  The algorithm and complexity of approximating the permanent of a matrix is an
  extensively studied topic.
  Recently, its connection with quantum supremacy and more specifically
  BosonSampling draws a special attention to the average-case approximation
  problem of the permanent of random matrices with zero or small mean value for
  each entry.
  Eldar and Mehraban (FOCS 2018) gave a quasi-polynomial time algorithm for
  random matrices with mean at least $1/\polyloglog (n)$.
  In this paper, we improve the result by designing a deterministic
  quasi-polynomial time algorithm and a PTAS for random matrices with mean at least
  $1/\polylog(n)$.
  We note that if it can be further improved to $1/\poly(n)$, it will disprove a
  central conjecture for quantum supremacy.

  Our algorithm is also much simpler and has a better and flexible trade-off for running time.
  The running time can be quasi-polynomial in both $n$ and $1/\eps$, or PTAS (polynomial in $n$ but
  exponential in $1/\eps$), where $\eps$ is the approximation parameter.
\end{abstract}

\section{Introduction}
\label{sec:intro}

The computational complexity of computing the permanent of a matrix is of
central importance to complexity theory and has been extensively studied ever
since Valiant's seminal result~\cite{valiant1979complexity}.
On one hand, the problem is algebraic in nature and plays an important role in
the study of algebraic
complexity~\cite{valiant1979completeness,burgisser2013algebraic}.
In particular, its relation with the determinant is an important
topic~\cite{mignon2004quadratic,cai2010quadratic}.
On the other hand, it also exhibits rich combinatorial properties.
The permanent can be viewed as counting the (weighted) number of perfect
matchings of a bipartite graph and graph (perfect) matching is one of the most
important graph problems in the study of algorithm and
complexity~\cite{edmonds1965maximum,valiant1979complexity,valiant2008holographic}.

Since the exact computation of the permanent is already \class{\textrm{\#}P}
hard for matrices with non-negative entries, or even $0/1$
entries~\cite{valiant1979complexity}, more recent research focuses on either the
\emph{approximation} of the permanent or the \emph{average-case} complexity of
the problem where the matrices are sampled from certain distributions.

In the approximation approach, we require the algorithm to return a number $Z'$
such that, if the actual value of the permanent of the input matrix is $Z$, the
computed number $Z'$ satisfies $\abs{Z - Z'} \le \eps \abs{Z}$ within
running time
$\poly(n, \frac{1}{\eps})$ where $\eps>0$ is the approximation parameter.
This is called a fully polynomial-time approximation scheme (FPTAS).
And its randomized relaxation is called a fully polynomial-time randomized
approximation scheme (FPRAS), where we require that $\abs{Z - Z'} \le \eps \abs{Z}$
holds with high probability.
If the running time is quasi-polynomial in terms of $n$ and $\frac{1}{\eps}$,
namely $2^{\poly(\log (n),\, \log \frac{1}{\eps})}$, then it is called a
quasi-polynomial time approximation scheme.
If we only require the running time to be polynomial in $n$ but not necessary in
$\frac{1}{\eps}$, we call it polynomial-time approximation scheme (PTAS).
On the other hand, in the average-case approach,
we allow the algorithm to be
incorrect on a small fraction of instances with respect to some distributions
over matrices.
Usually, this distribution is over matrices with \iid\ random
entries and the algorithm is required to output either the exact value or an
approximation of the permanent on at least $1 - o(1)$ fraction of the
instances.

In fact, several worst-case approximation tractability and hardness results were known.
For a matrix with non-negative entries, Jerrum, Sinclair and Vigoda gave a
remarkable FPRAS to approximate its permanent~\cite{jerrum2004polynomial} via
random sampling by Markov chain Monte Carlo (MCMC).
How to derandomize this algorithm remains a long-standing open problem.
However, it is impossible to extend this result to general matrices
since it is already
\class{\textrm{\#}P}-hard to compute the sign of the permanent with possibly
negative entries.
Indeed, negative or complex values put this problem in \class{GapP}, a
superset of \class{\textrm{\#}P}~\cite{FENNER1994}.
This difficulty is referred to as the ``interference barrier''.
For example, random sampling based algorithms are no longer applicable since we
cannot define negative or complex probability.
For specific families of complex matrices, there are quasi-polynomial time
approximation schemes by Barvinok's
interpolation~\cite{barvinok2016computing,barvinok2019computing}.

The above algorithms and hardness results are all worst-case analysis.
What do we know about the average-case complexity?
It turns out that, for exact counting, the average-case problem remains
\class{\textrm{\#}P}-hard both for finite field entries and complex Gaussian
entries~\cite{CaiPS99,aaronson2013computational}.
This leaves the complexity of the average-case approximation of the permanent an
intriguing problem.
Yet, very little was known when both average-case analysis and approximation
are considered at the same time.
More generally, while we have \class{\textrm{\#}P}-hardness results for all
other settings including worst-case approximation problems and average-case
exact problems, essentially no hardness result is known for average-case
approximate counting problems.
On the tractability side, some recent algorithms and techniques show that random
instances might be much easier than the worst-case for approximate
counting~\cite{jenssen2019algorithms,LiaoLLM19,mann2019approximation,eldar2018approximating}.
Our result also adds to this list.

An important motivation for studying the complexity of approximating the
permanent of random matrices stems from the so-called BosonSampling program
initiated by Aaronson and Arkhipov~\cite{aaronson2013computational} in quantum
computing.
In~\cite{aaronson2013computational}, the conjectured
\class{\textrm{\#}P}-hardness of approximating the permanent of Gaussian
matrices (Permanent-of-Gaussian) is connected to the sampling problem of linear
optical networks so that the existence of any efficient classical simulation of
this optical sampling process will imply $\class{P}^{\class{\textrm{\#} P}}
{=}\, \class{BPP}^{\class{NP}}$, and hence the collapse of the polynomial
hierarchy by Toda's theorem.
This provides an explicitly defined problem which near-term quantum computing
devices can efficiently solve while even today's most powerful classical
supercomputer cannot.
Such a dramatic contrast in computing powers, called quantum supremacy, poses
the first serious challenge to the extended Church-Turing thesis and has been
recently experimentally achieved by the Google team using a different model
based on the random quantum circuit sampling problem~\cite{arute2019quantum}
while the record of BosonSampling experiment is recently updated
by~\cite{wang2019boson}.

The complexity of approximating the permanent of random matrices is obviously of
vital importance to BosonSampling as it serves as one of the two conjectures on
which the theory of BosonSampling bases.
In particular, it is assumed in BosonSampling that approximating the permanent
of random matrices whose entries are \iid\ sampled from the normal distribution
of zero mean value and unit variance is \class{\textrm{\#}P} hard.
The result is strengthened in~\cite{eldar2018approximating} showing that a
biased distribution with mean at most $1/\sqrt{n}$ is also good enough for
BosonSampling.
There is no clue yet on how one can prove such hardness results and it is not
even clear whether they are true or not.
On the other side, a surprising and interesting tractable result was obtained by
Eldar and Mehraban~\cite{eldar2018approximating}.
They provided a quasi-polynomial time algorithm to approximate the permanent of
random matrices with mean of $1/\polyloglog (n)$, which implies that the
\class{\textrm{\#}P}-hardness is unlikely to hold for these families of random
matrices.
This raises the interesting open question of whether the algorithm can be
extended to the case with mean value $1/\poly(n)$
and disprove the hardness
conjecture of BosonSampling or there is a ``phase transition'' in the complexity
of approximating the permanent with respect to the mean of matrix entries.

\subsection{Our results}

In this paper, we provide an exponential improvement in terms of the tractable
region of the mean values to the problem of approximating the permanent of
random matrices with vanishing mean value.
We design a deterministic quasi-polynomial time algorithm and a PTAS that can compute the
multiplicative approximation for $1-o(1)$ fraction of random matrices with mean
at least $1/\polylog (n)$.
See \cref{thm:main}, \cref{cor:main}, \cref{cor:main2} for more rigorous statements of our results.
The strength of our results lies in the following four aspects.

Firstly and most importantly, the range of the tractable mean value parameters
is exponentially better than that of~\cite{eldar2018approximating}.
In~\cite{eldar2018approximating}, their algorithm can only approximate the
permanent of a random matrix with mean value of at least $1/\polyloglog (n)$.
Our algorithm works for all mean value that is $1/\polylog (n)$.
The exact range of mean values for which such approximation exists is extremely
important due to its role in the ``quantum supremacy''.
If one can further improve the mean to $1/\poly(n)$, it will disprove the
conjecture in~\cite{aaronson2013computational}.

Secondly, the algorithm in~\cite{eldar2018approximating} only works for some,
but not all, mean values $\mu > 1/\polyloglog (n)$.
This is a very strange situation due to their proof techniques and is rather
counterintuitive as one would expect that the larger the mean value is,
the easier it is to approximate the permanent.
There is not even an algorithm for them to check whether a given mean is
computable or not for their algorithm since they used a probabilistic argument
while our algorithm does not suffer from such problems and works for all $\mu >
1/\polylog (n)$.

Thirdly, our algorithm uses a completely different idea and is arguably simpler.
The simplicity of our algorithm also enables us to extend our result to a
much larger family of entry distributions.
While the technique of~\cite{eldar2018approximating} is very interesting and
uses Barvinok's interpolation method for approximate counting with several new
developments of the technique in a few directions, the need to prove the
zero-freeness of the polynomial in a segment-like region made the proof rather
involved and caused the drawbacks of their result mentioned above in the
previous two items.
In our algorithm, we avoid all these complications and simply truncate a simple
expansion of the permanent directly to
$\mathcal{O}(\ln n + \ln \frac{1}{\epsilon})$ terms and compute
them by brute-force.
We note that in usual usage of Barvinok's method, a truncation of the
polynomial directly rather than its logarithm will not succeed.
To prove the correctness of our algorithm, we need a careful study of the
distribution of the permanent of random matrices.
To this end, we use several techniques inspired by the analysis
of~\cite{rempala2004limit}.
The result and analysis in~\cite{rempala2004limit} are asymptotical while we
need much more careful quantitative bounds which we develop carefully in this
paper.

Lastly, the running time of our algorithm is also better and flexible.
It can be quasi-polynomial in both $n$ and $1/\eps$, where $\eps$ is the approximation parameter.
We can also make it a PTAS, which is polynomial in $n$ but exponential in $1/\eps$.
In particular in the range of $\eps>n^{-\rho}$ for some constant $\rho>0$, the algorithm
is extremely simple and runs in only linear time. It is not clear how to make the previous
algorithm in polynomial time rather than quasi-polynomial time  even for a fixed constant $\eps$.

This work leaves several interesting open problems.
First, the most important problem left open is to either show the transition of
complexity with respect to the mean value and prove that the corresponding
problem is hard when the mean value is $1/\poly(n)$ or disprove the
Permanent-of-Gaussian conjecture of BosonSampling.
With our technique only, it is rather hard to go beyond the $1/\polylog (n)$
barrier and essential new ideas seem necessary if this is ever possible.
Second, while we have been focusing exclusively on the problem of approximating
the permanent and therefore it is only directly relevant to the BosonSampling
scheme of quantum supremacy, we expect that our technique may find applications
in understanding other average-case approximate counting problems and the
hardness assumptions in other quantum supremacy schemes such as the
instantaneous quantum computing model~\cite{bremner2010classical} and the random
circuit sampling model~\cite{arute2019quantum}.
We believe that such generalizations are possible as the hardness conjectures
behind different models of quantum supremacy are of the same flavor and it is
usually possible to generalize results from one model to the
other (see e.g.~\cite{bouland2018quantum}).

The rest of the paper is organized as follows.
In \cref{sec:prelim}, we introduce the notations used in this paper.
We state and present the proof outline in \cref{sec:main}.
The remaining sections contain the technical lemmas used in \cref{sec:main}.

\section{Preliminary}
\label{sec:prelim}

In this paper we use $[n]$ to denote the set $\{1, \cdots, n\}$.
The set of natural numbers, real numbers, and complex numbers are denoted as
$\N$, $\R$, and $\C$ respectively.
We use $n^{\underline{k}} \triangleq n (n - 1) \cdots (n - k + 1)$ to denote the
downward factorial and $C_{n, k}, P_{n, k}$ to denote all $k$-subsets of $[n]$
and all $k$-permutations of $[n]$ respectively.
$\Matrix_n(\C)$ to denote the set of all $n \times n$ complex matrices.
$\delta_{i, j}$ is the Kronecker function, i.e., $\delta_{i, j} = 1$ if $i = j$
and $\delta_{i, j} = 0$ otherwise.

\begin{definition} \label{def:s-e-poly}
  Suppose $x_1, x_2, \cdots, x_n \in \C$ and $0 \le k \le n$,
  the power sum is defined by
  \[
    S_k(x_1, x_2, \cdots, x_n) \triangleq \sum_{i = 1} ^ n x_i^k
  \]
  and the $k$th elementary symmetric polynomial is defined by
  \[
    e_k(x_1, x_2, \cdots, x_n) \triangleq
    \sum_{\{i_1, \cdots, i_k\} \in C_{n, k}} x_{i_1} \cdots x_{i_k}
  \]
  with convention $e_0(x_1, x_2, \cdots, x_n) = 1$.
  We will write $S_k(n)$ and $e_k(n)$ if the variables $x_i$'s are clear from
  context.
\end{definition}

\begin{definition}
  \label{def:entry-distr}
  The entry distribution $\mathcal{D}_\mu$ with mean value $\mu$ is a
  distribution over complex numbers such that
  \begin{equation*}
    \Ex_{x \sim \mathcal{D}_\mu} [x] = \mu,\quad
    \Var_{x \sim \mathcal{D}_\mu} [x] = 1,
  \end{equation*}
  and
  \begin{equation*}
    \Ex_{x \sim \mathcal{D}_\mu} \abs{x - \mu}^3 = \rho < \infty.
  \end{equation*}
  We use $\mathcal{D}$ to denote $\mathcal{D}_0$.
\end{definition}

In this paper, we use $\xi$ to denote the quasi-variance of $\mathcal{D}_\mu$,
\begin{equation*}
  \xi = \Ex_{x \sim \mathcal{D}_\mu} (x-\mu)^2.
\end{equation*}
The norm of the quasi-variance is upper bounded by the variance as
\begin{equation}
  \label{eq:xi-bound}
  \begin{split}
    \abs{\xi}\, & =\, \abs{\Ex_{x \sim \mathcal{D}_\mu} (x-\mu)^2}\\
    & \le\, \Ex_{x \sim \mathcal{D}_\mu} \abs{x-\mu}^2\\
    & =\, \Var_{x \sim \mathcal{D}_\mu} (x) = 1.
  \end{split}
\end{equation}

\begin{definition}
  \label{def:matrix-distr}
  The matrix distribution $\mathcal{M}_{n,\, \mu}$ is the distribution over
  $\bm{R} \in \Matrix_n(\C)$ such that the entries of $\bm{R}$ are \iid\ sampled
  from $\mathcal{D}_\mu$.
\end{definition}

Our aim is to design an average-case approximation algorithm for the permanent
of a random matrix $\bm{R} \sim \mathcal{M}_{n,\, \mu}$ for $\mu =
\polylog^{-1}(n)$.
Following the notation used in~\cite{eldar2018approximating}, we introduce a
matrix $\bm{X} = \bm{J} + z \bm{A}$ where $z$ is a complex variable taken to be
$1/\mu$ in the end, $\bm{J}$ is the all-ones matrix, and $\bm{A}$ is a random
matrix with \iid\ entries sampled from $\mathcal{D}$.
We note that $\frac{\bm{X}}{z} \sim \mathcal{M}_{n,\, \mu}$ with $z=1/\mu$ and
thus it is equivalent to compute the permanent of matrix $\bm{X}$.

\begin{definition}
  Suppose $\bm{A} \in \Matrix_n(\C), \bm{B} \in \Matrix_k(\C)$.
  We write $\bm{B} \subseteq_k \bm{A}$ if $\bm{B}$ is a $k \times k$ submatrix
  of $\bm{A}$.
\end{definition}

\begin{lemma}
  \label{lem:perm-expansion}
  Suppose $\bm{A} \in \Matrix_n(\C)$ is any matrix and $\bm{J}$ is the
  all-ones matrix of size $n$.
  For $k=0, 1, \ldots, n$, define
  \begin{equation}
    \label{eq:a_k}
    a_k = \frac{1}{n^{\underline{k}}}\, \sum_{\bm{B} \subseteq_k \bm{A}}
    \Per(\bm{B}),
  \end{equation}
  Then for all $z \in \C$ we can write
  \begin{equation*}
    \frac{\Per(\bm{J} + z \bm{A})}{n!} = \sum_{k=0}^n a_k z^k.
  \end{equation*}
\end{lemma}

\begin{proof}
  This identity can be obtained by simple calculate and similar formula has
  appeared in~\cite{eldar2018approximating,rempala2004limit}.
  Given any $n \times n$ matrix $\bm{M}$, define $G_{\bm{M}}$ to be the
  corresponding complete bipartite graph with $n$ vertices on each side, both
  numbered from 1 to n, where the weight of edge $e = (i, j)$ is simply
  $\bm{M}_{i, j}$.
  Define the weight of any perfect matching in $G_{\bm{M}}$ to be the product of
  weights of all edges in it.
  By definition, permanent of $\bm{M}$ can be seen as weights of all perfect
  matchings of $G_{\bm{M}}$.
  Ideally, we can split any perfect matching of $G_{\bm{J} + z\bm{A}}$ into
  combinations of a $k$-matching of $G_{\bm{J}}$ and a $(n - k)$-matching of
  $G_{z\bm{A}}$ for some $k$ such that the two matchings together form a perfect
  matching.
  Regarding $\Per(\bm{J} + z\bm{A})$ as a polynomial of $z$, the coefficient of
  $z^k$ is simply the summation of weights over all $k$-matchings of $G_{\bm{A}}$
  times the weights of all perfect matchings in the left graph of $G_{\bm{J}}$,
  which is $(n - k)!$.
  Thus the lemma follows by rescaling.
\end{proof}

We record here some basic inequalities that we use extensively in the proofs
\begin{equation}
  \begin{aligned}
    n! \ge \left(\dfrac{n}{e}\right)^n \quad \forall n \in \N \,, \\
    (1 + x)^y \le e^{xy} \quad \forall x, y > 0\,.
  \end{aligned}
\end{equation}
And we always assume $0^0 = 1$ in this paper.

\section{Main Result}
\label{sec:main}

In this section, we state our main result and describe the overall proof
structures.
The key technical lemmas are proved in the later sections.

\begin{theorem}
  \label{thm:main}
  For any constant $c \in (0, \frac{1}{8})$, there exists a deterministic
  quasi-polynomial time algorithm $\mathcal{P}$ such that, given both a matrix
  $\bm{R}$ sampled from $\mathcal{M}_{n,\, \mu}$ defined in
  \cref{def:matrix-distr} for $\abs{\mu} \ge \ln^{-c} (n)$ and a real number
  $\eps \in (0,1)$ as input, the algorithm computes in time $n^{\mathcal{O}(\ln
    n + \ln \frac{1}{\eps})}$ a complex number $\mathcal{P}(\bm{R}, \eps)$ that
  approximates the permanent $\Per(\bm{R})$ on average in the sense that
  \begin{equation*}
    \Pr \biggl( \abs{ 1 - \frac{\mathcal{P}(\bm{R}, \eps)}{\Per(\bm{R})}}
    \le \eps \biggr) \ge 1 - o(1),
  \end{equation*}
  where the probability is over the random matrix $\bm{R}$.
\end{theorem}

\begin{proof}
  As discussed in \cref{sec:prelim}, we will work with the permanent of matrix
  $\bm{X} = \bm{J} + z\bm{A}$ where $\bm{J}$ is the all-ones matrix.
  In the following, we design an algorithm that can approximate
  $\Per(\bm{X})$ on average for $\bm{A} \sim \mathcal{M}_n$.
  The algorithm $\mathcal{P}$ and its performance then follow by a simple
  scaling argument.

  Since $\Per(\bm{X})$ is a summation of $n!$ products, it is convenient to
  focus on computing the normalized permanent $\frac{\Per(\bm{X})}{n!}$, which
  can be written as $\bigsum_{k=0}^n a_kz^k$ by \cref{lem:perm-expansion} for
  \begin{equation*}
    a_k = \frac{1}{n^{\underline{k}}} \sum_{\bm{B} \subseteq_k \bm{A}} \Per(\bm{B}).
  \end{equation*}

  We first fix the parameters used commonly
  in the rest of the proof to be any solution of the following equalities.
  \begin{equation}
    \label{main-parameters}
    \begin{cases}
      0 < c < \nu < \frac{1}{8}, \\
      0 < \gamma < \beta < \frac{1}{2}, \\
      0< \gamma < \nu - c, \\
      \abs{z} \le (\ln n)^c, \\
      t = \ln n + \ln \frac{1}{\eps}, \\
      \theta(n) = \ln\ln n.
    \end{cases}
  \end{equation}
  Note that $z$ can be a complex number.

  To approximate the normalized permanent of $\bm{X}$, our algorithm computes
  the first $\cut+1$ coefficients $a_0, a_1, \ldots, a_\cut$
  and outputs the number $\sum_{k=0}^\cut a_k z^k$.
  For such choice of $\cut$, the algorithm has time complexity
  $\mathcal{O}(n^{2\cut} \cdot \poly(\cut) \cdot \cut!)$,
  which is quasi-polynomial.
  The rest of the proof is to show that the first $\cut + 1$ terms in the
  summation is actually an $\Theta(\eps)$ approximation of
  $\frac{\Per(\bm{X})}{n!}$.

  The easy part is to prove that the remaining terms are indeed small with high
  probability. Namely, w.p. at least $1 - o(1)$,
  \begin{equation}
    \label{eq:a_k-tail-bound}
    \left\vert \bigsum_{k = \cut + 1}^n a_k z^k \right\vert
    \le n^{-\gamma}\, \eps
  \end{equation}
  This is small in absolute sense.
  To show that it is small relatively, we need to give a lower bound of $\sum_{k
    = 0}^\cut a_k z^k$. Namely, w.p. at least $1 - o(1)$,
  \begin{equation}
    \label{eq:a_k-large}
    \left\vert \sum_{k = 0}^\cut a_k z^k \right\vert \ge \Theta(n^{-\gamma}).
  \end{equation}
  As the constant in $\Theta(n^{-\gamma})$ does not depend on $\eps$, these two
  facts together give a proof of the main theorem.

  The former fact in \cref{eq:a_k-tail-bound} is proven in
  \cref{lem:a_k-tail-bound}.
  The later bound in \cref{eq:a_k-large} is more difficult
  because we do not know how to bound $a_k$'s directly.
  To overcome this, we use symmetric polynomials of the column sum of matrix
  $\bm{A}$ to approximate the permanent.

  For all $j=1, 2, \ldots, n$, define
  \begin{equation}
    \label{eq:column-sum}
    C_j \triangleq \frac{1}{\sqrt{n}} \sum_{i = 1}^n a_{i,j},
  \end{equation}
  where $a_{i,j}$ is the $(i,j)$-the entry of matrix $\bm{A}$ and
  for all $k=0,1, 2, \ldots, n$,
  \begin{align}
    \label{eq:V_k}
    V_k & \triangleq \frac{1}{n^{k/2}}\, e_k (C_1, C_2, \ldots, C_n),\\
    \label{eq:D_k}
    D_k & \triangleq \frac{1}{n^{k/2}}\, S_k (C_1, C_2, \ldots, C_n),
  \end{align}
  where polynomials $S_k$ and $e_k$ are defined in \cref{def:s-e-poly}.

  The first two terms are exactly the same, namely $V_0=a_0$ and $V_1=a_1$.
  Consider $C_j$ and $V_k$ as multivariable polynomials of $a_{i,j}$'s.
  Intuitively, as $n \to \infty, V_k$ and $a_k$ share almost all monomials, thus
  nearly of the same value.
  In particular, $V_k$ is an unbiased estimation of $a_k$.
  Formally, we prove in \cref{lem:sigma-ak-sigma-vk-difference} that w.p.
  $1 - o(1)$,
  \begin{equation}
    \label{eq:sigma-ak-sigma-vk-difference}
    \left\vert \bigsum_{k = 0}^\cut a_k z^k - \bigsum_{k = 0}^\cut V_k z^k
    \right\vert \le n^{-\beta} = o(n^{-\gamma}),
  \end{equation}
  which is negligible compared to the target $n^{-\gamma}$.

  As proven in \cref{lem:product-sum-formula}, $V_k$ satisfy the so-called
  Newton's identities.
  \begin{equation}
    \label{eq:v-k-recursion}
    V_k = \frac{V_{k - 1} V_1 - V_{k - 2} D_2 + \sum_{i = 2}^{k - 1}(-1)^i
      V_{k-1-i} D_{i+1}}{k} \quad\text{for all } k \ge 2.
  \end{equation}

  Further more, we prove in \cref{lem:bound-d2} that $D_2$ is
  concentrated at $\xi$, the quasi-variance of $\mathcal{D}$, and
  in \cref{lem:bound-dk} that $D_k$ is
  polynomially small for $k\geq 3$ both with high probability.
  This motivates us to consider $V'_k$, an asymptotic approximation of $V_k$, as
  follows
  \begin{equation}
    \label{eq:v-k'-recursion}
    V_k' =
    \begin{cases}
      1, & k = 0, \\
      V_1, & k = 1, \\
      \dfrac{V_{k-1}'V_1' - V_{k-2}'\xi}{k}, & k \ge 2.
    \end{cases}
  \end{equation}
  Note that $k$ can be larger than $n$ for notation convenience when analyzing
  $\sum_{k=0}^\infty V_k' z^k$.
  And we prove in \cref{lem:eps-bound} that
  $V_k$ and $V'_k$ are close and in \cref{lem:v-k-v-k'-sum} that
  w.p. $1-o(1)$,
  \begin{equation}
    \label{eq:v-k-v-k'-sum}
    \left\vert \bigsum_{k = 0}^\cut V_k\, z^k - \bigsum_{k = 0}^\cut V'_k\, z^k
    \right\vert = \mathcal{O}(n^{c - \nu}) = o(n^{-\gamma}).
  \end{equation}

  Comparing the two recursions, we use the ``probabilists' Hermite polynomials''
  to explicit express $(V_k')$'s in \cref{eq:vc'-formula}.
  Due to \cref{lem:v-k-summation-tail}, the summation $\sum_{k = 0}^\cut V'_k
  z^k$ can be estimated by $\sum_{k = 0}^\infty V_k' z^k$ with a negligible
  $n^{-\omega(1)}$ additive error by an upper-bound of ``probabilists' Hermite
  polynomials'' in \cref{lem:hermite-bound}.
  This, together with \cref{eq:sigma-ak-sigma-vk-difference,eq:v-k-v-k'-sum},
  implies that it is enough to give a $\Omega(n^{-\gamma})$ lower-bound of
  $\abs{\sum_{k = 0}^{\infty} V_k' z^k}$.

  On the other hand, from \cref{eq:v-k'-summation}, $\sum_{k = 0}^\infty V_k'
  z^k$ is simply $e^{V_1 z - \frac{\xi z^2}{2}}$, where $V_1$ is the normalized
  average of all entries in $\bm{A}$.
  By Chebyshev's inequality, we know that $V_1$ is small with high probability.
  This can be used to prove the fact in \cref{lem:v'-summation-big} that w.p.
  $1-o(1)$,
  \[
    \abs{ e^{V_1 z - \frac{\xi z^2}{2}} } \ge n^{-\gamma},
  \]
  which completes the proof.
\end{proof}

If we relax the approximation requirement a bit, we can simply compute $V_1$
and return $n! e^{V_1 z - \frac{\xi z^2}{2}}$ as an approximation of $\Per(\bm{X})$.
This is a truly polynomial time algorithm and extremely simple. By the above argument,
we can get the following approximation guarantee.

\begin{corollary}
  \label{cor:main}
  For any constant $c \in (0, \frac{1}{8})$ and $0< \rho < \frac{1}{8}-c $, there
  exists a deterministic polynomial time algorithm $\mathcal{P}$ such that,
  given a matrix $\bm{R}$ sampled from $\mathcal{M}_{n,\, \mu}$ defined in
  \cref{def:matrix-distr} for $\abs{\mu} \ge \ln^{-c} (n)$, the algorithm
  computes a complex number $\mathcal{P}(\bm{R})$ that approximates the
  permanent $\Per(\bm{R})$ on average in the sense that
  \begin{equation*}
    \Pr \biggl( \abs{ 1 - \frac{\mathcal{P}(\bm{R})}{\Per(\bm{R})}}
    \le n^{-\rho} \biggr) = 1 - o(1),
  \end{equation*}
  where the probability is over the random matrix $\bm{R}$.
\end{corollary}

The drawback of this simple algorithm is that we do not have a parameter $\eps$
to control the approximation precision.
However, for large $n$ this is already a very good approximation algorithm while
for small $n$ we can just compute $\Per(\bm{R})$ directly by Ryser formula in time $2^n$.
By this idea, we can convert the above algorithm into a PTAS but not FPTAS,
whose running time is polynomial in $n$ but possibly exponential in
$\frac{1}{\eps}$.
Let us fix a constant $0< \rho < \frac{1}{8}-c $.
For a given approximation parameter $\eps$, if $\eps> n^{-\rho} $ we use the
above polynomial time algorithm and otherwise simply compute it directly.
The running time is bounded by $\max \Bigl\{ \poly(n),
2^{\eps^{-\frac{1}{\rho}}} \Bigr\}$, which shows that the modified algorithm is
a PTAS.

\begin{corollary}
  \label{cor:main2}
  For any constant $c \in (0, \frac{1}{8})$, there exists a deterministic PTAS
  to approximate $\Per(\bm{R})$ for $1 - o(1)$ fraction of random matrices
  $\bm{R}$ sampled from $\mathcal{M}_{n,\, \mu}$ defined in
  \cref{def:matrix-distr} for $\abs{\mu} \ge \ln^{-c} (n)$.
\end{corollary}

\section{Estimation with Summation of Columns}

In this section, we prove that the two summations $\sum_{k = \cut + 1}^n a_k z^k$
and $\sum_{k = 0}^\cut (a_k-V_k) z^k$ are both small with high probability.
Their proofs are similar and simply follow from the fact that they have zero
mean and exponentially decaying variance.

Recall from \cref{eq:a_k,eq:V_k} that
\begin{equation*}
  a_k = \frac{1}{n^{\underline{k}}}\, \sum_{\bm{B} \subseteq_k \bm{A}}
  \Per(\bm{B}),\quad
  V_k = \dfrac{1}{n^{k/2}} \bigsum_{\{j_1, \cdots, j_k\} \in C_{n, k}} C_{j_1}
  \cdots C_{j_k}.
\end{equation*}

\begin{lemma}
  \label{lem:ak-mean-variance}
  For any $k, \ell \in \{ 0, 1, \ldots, n \}$,
  \begin{equation*}
    \Ex[a_k] = \delta_{k,0}\,,\quad
    \Ex[a_k \,\overline{a_\ell}] = \frac{\delta_{k,\ell}}{k!}.
  \end{equation*}
\end{lemma}

\begin{proof}
  As $a_0 \equiv 1$, it holds that $\Ex[a_0] = 1$. For $k > 0$,
  \begin{equation}
    \label{eq:Ex-ak}
    \begin{split}
      \Ex[a_k] & = \frac{1}{n^{\underline{k}}} \sum_{\bm{B} \subseteq_k \bm{A}}
      \Ex \!\left[\Per(\bm{B}) \right]\\
      & = \frac{1}{n^{\underline{k}}}
        \bigsum_{\substack{\{i_1, \cdots, i_k\} \in C_{n, k} \\
          \{j_1, \cdots, j_k\} \in C_{n, k}}}
        \Ex\! \left[\bigsum_{\sigma \in P_{k,k}}
        \bigprod_{t=1}^k a_{i_t, j_{\sigma_t}} \right]\\
      & = \frac{1}{n^{\underline{k}}}
        \bigsum_{\substack{\{i_1, \cdots, i_k\} \in C_{n, k} \\
          (j_1, \cdots, j_k) \in P_{n, k}}}
        \Ex\!\left[\prod_{t=1}^k a_{i_t, j_t} \right].
    \end{split}
  \end{equation}
  By the fact that the entries of $\bm{A}$ are \iid\ and have $0$ mean value, the
  above expectation is $0$. This proves $\Ex[a_k] = 0$ for $k > 0$.

  For the second part, we compute
  \begin{align*}
    \Ex[a_k \overline{a_\ell}]
    = & \frac{1}{n^{\underline{k}}\, n^{\underline{\ell}}}\,
      \sum_{\bm{B} \subseteq_k \bm{A}} \sum_{\bm{B}' \subseteq_\ell \bm{A}}
      \Ex\!\left[ \Per(\bm{B}) \overline{\Per(\bm{B}')} \right]. \\
    = & \dfrac{1}{n^{\underline{k}}\, n^{\underline{\ell}}}\,
      \bigsum_{\substack{\{i_1, \cdots, i_k\} \in C_{n, k} \\
        (j_1, \cdots, j_k) \in P_{n, k}}}
      \bigsum_{\substack{\{i_1', \cdots, i_l'\} \in C_{n, l} \\
        (j_1', \cdots, j_l') \in P_{n, l}}}
      \Ex\!\left[ \bigprod_{t = 1}^k a_{i_t, j_t}
      \bigprod_{t = 1}^k \overline{a_{i_t', j_t'}} \right] \\
  \end{align*}

  Since all entries $\bm{A}$ are iid and of zero mean, the expectation of
  two products is non-zero only if the corresponding subscripts in the two
  products are equal, i.e.
  \[
    \{(i_1, j_1), \cdots, (i_k, j_k)\} = \{(i_1', j_1'), \cdots, (i_l', j_l')\}.
  \]
  This proves that $\Ex[a_k \overline{a_\ell}]=0$ for $k \ne \ell$ and also
  simplifies $\Ex[a_k\overline{a_k}]$ as
  \begin{equation*}
    \Ex[a_k \overline{a_k}]
    = \dfrac{1}{(n^{\underline{k}})^2}\,
      \bigsum_{\substack{\{i_1, \cdots, i_k\} \in C_{n, k} \\
        (j_1, \cdots, j_k) \in P_{n, k}}}
      \Ex\!\left[ \bigprod_{t = 1}^k \abs{a_{i_t, j_t}}^2 \right]
    = \frac{1}{(n^{\underline{k}})^2}\, \binom{n}{k}^2 k! = \frac{1}{k!},
  \end{equation*}
  because entries are iid and the variance of each entry is 1.
\end{proof}

\begin{lemma}
  \label{lem:vk-mean-variance}
  For any $k \in \{0, 1, \cdots, n\}$,
  \begin{equation*}
    \Ex[V_k] = \delta_{k,0}, \quad
    \Ex\!\abs{V_k - a_k}^2 \le \dfrac{k(k-1)}{2 n \cdot k!}.
  \end{equation*}
\end{lemma}
\begin{proof}
  By definition, $V_0 = a_0 \equiv 1$, which proves the lemma for $k = 0$.
  For positive $k$, similar to \cref{lem:ak-mean-variance}, we have
  \begin{equation}
    \label{eq:Ex-Vk}
      \Ex (V_k)
      = \frac{1}{n^{k/2}}
        \bigsum_{\{j_1, \cdots, j_k\} \in C_{n, k}}
        \Ex[C_{j_1} \cdots C_{j_k}]
      = \frac{1}{n^k}
        \bigsum_{\substack{\{j_1, \cdots, j_k\} \in C_{n, k} \\
          i_1, \cdots, i_k \in [n]}}
          \Ex\left[\prod_{\ell=1}^k a_{i_\ell, j_\ell}\right].
  \end{equation}

  By a similar argument in the proof for \cref{lem:ak-mean-variance},
  \begin{equation*}
    \begin{split}
      \Ex\!\abs{V_k - a_k}^2
      & = \Ex \left\vert
        \bigsum_{\substack{\{j_1, \cdots, j_k\} \in C_{n, k} \\
          (i_1, \cdots, i_k) \in P_{n, k}}}\,
        \prod_{t=1}^k a_{i_t,\, j_t} \left(
        \frac{1}{n^k} - \frac{1}{n^{\underline{k}}} \right)
        + \bigsum_{\substack{\{j_1, \cdots, j_k\} \in C_{n, k} \\
          (i_1, \cdots, i_k) \in [n] ^ k - P_{n, k}}}\,
          \prod_{t=1}^k a_{i_t, j_t} \cdot \frac{1}{n^k} \right\vert^2\\
      & = \bigsum_{\substack{\{j_1, \cdots, j_k\} \in C_{n, k} \\
        (i_1, \cdots, i_k) \in P_{n, k}}}\,
        \left(\frac{1}{n^k} - \frac{1}{n^{\underline{k}}}
        \right)^2 \bigprod_{t=1}^k \Ex\!\abs{a_{i_t,\, j_t}}^2
        + \bigsum_{\substack{\{j_1, \cdots, j_k\} \in C_{n, k} \\
          (i_1, \cdots, i_k) \in [n] ^ k - P_{n, k}}}\,
          \frac{1}{n^{2k}} \bigprod_{t=1}^k \Ex\!\abs{a_{i_t,\, j_t}}^2\\
      & = \left( \frac{1}{n^k} - \frac{1}{n^{\underline{k}}} \right)^2
        \binom{n}{k}\, n^{\underline{k}} + \binom{n}{k}\,
        \frac{n^k - n^{\underline{k}}}{n^{2k}}
      = \frac{n^k - n^{\underline{k}}}{k!\, n^k}.
    \end{split}
  \end{equation*}
  The bound on $\Ex\!\abs{V_k - a_k}^2$ then follows by observing
  \begin{equation*}
      \frac{n^k - n^{\underline{k}}}{k!\, n^k}
      = \frac{1}{k!} \left[ 1-\prod_{t=0}^{k-1}
        \left( 1-\frac{t}{n} \right) \right]
      \le \frac{1}{k!} \sum_{t=0}^{k-1}\frac{t}{n} = \frac{k(k-1)}{2n\cdot k!}.
  \end{equation*}
  \qedhere
\end{proof}

Then we prove \cref{eq:a_k-tail-bound,eq:sigma-ak-sigma-vk-difference} as
follows.
\begin{lemma}
  \label{lem:a_k-tail-bound}
  With parameters satisfying \cref{main-parameters}, it holds that
  \[
    \Pr \left( \abs{ \sum_{k = \cut + 1}^n a_k z^k }
      \le n^{-\gamma}\, \eps \right) \ge 1-o(1).
  \]
\end{lemma}
\begin{proof}
  To apply Chebyshev's inequality, we first calculate the variance of
  $\bigsum_{k = \cut + 1}^n a_k z^k$.
  By \cref{lem:ak-mean-variance}, we have
  \begin{equation*}
    \Ex \left[ \sum_{k = \cut + 1}^n a_k z^k \right] =
    \sum_{k = \cut + 1}^n \Ex [a_k]\, z^k = 0,
  \end{equation*}
  and
  \begin{equation*}
    \Var \! \left(\sum_{k = \cut + 1}^n a_k z^k \right)
    =  \Ex  \abs{\sum_{k = \cut + 1}^n a_k z^k}^2
    =  \sum_{k = \cut + 1}^n \sum_{\ell = \cut + 1}^n
        \Ex (a_k \overline{a_\ell})\, z^{k} \overline{z}^\ell
    =  \sum_{k = \cut + 1}^n \frac{\abs{z}^{2k}}{k!}.
  \end{equation*}
  Applying Chebyshev's inequality, we have
  \begin{equation}
    \label{eq:ak-tail-1}
    \begin{split}
      \Pr\! \left[ \abs{\sum_{k = \cut + 1}^n a_k z^k}
        \ge n^{-\gamma}\, \eps \right]
      \le \dfrac{n^{2\gamma}}{\eps^2}
      \Var\! \left( \sum_{k = \cut + 1}^n a_k z^k \right)
      \le \dfrac{n^{2\gamma}}{\eps^2}
      \sum_{k = \cut + 1}^n \dfrac{\abs{z}^{2k}}{k!}.
    \end{split}
  \end{equation}
  As chosen in \cref{main-parameters},
  \begin{equation}
    \label{eq:cut-mu-1}
    \begin{cases}
      \abs{z} \le (\ln n)^{\frac{1}{8}} \\
      \cut = \ln n + \ln \frac{1}{\eps}.
    \end{cases}
  \end{equation}
  In this case, for any $k \ge \cut$ and large $n$,
  \begin{equation*}
    \dfrac{\abs{z}^{2(k + 1)}}{(k + 1)!} \div
      \dfrac{\abs{z}^{2k}}{k!}
    = \dfrac{\abs{z}^{2}}{k + 1}
    \le \dfrac{\abs{z}^2}{\cut + 1} < 1/2.
  \end{equation*}
  We can continue \cref{eq:ak-tail-1} as
  \begin{equation}
    \begin{split}
      \Pr\! \left[ \abs{\sum_{k = \cut + 1}^n a_k z^k}
        \ge n^{-\gamma}\,\eps \right]
      \le \dfrac{n^{2\gamma} \, \abs{z}^{2\cut}}{\eps^2 \, \cut!}.
    \end{split}
  \end{equation}
  Since $\cut = \ln n + \ln \frac{1}{\eps}$, it is clear
  that $\frac{t!}{\abs{z}^{2t}}$ is super-polynomial for large $n$, which means
  \begin{equation} \label{eq:a-tail}
    \Pr\!\left[\left\vert \bigsum_{k = \cut + 1}^n a_k z^k \right\vert
      \ge n^{-\gamma}\,\eps \right]
    = o(1).
  \end{equation}
\end{proof}

\begin{lemma}
  \label{lem:sigma-ak-sigma-vk-difference}
  With all parameters satisfying \cref{main-parameters},
  \begin{equation*}
    \Pr\!\left(\abs{ \sum_{k = 0}^\cut a_k z^k - \sum_{k = 0}^\cut V_k z^k}
    \le n^{-\beta} \right) = 1 - o(1).
  \end{equation*}
\end{lemma}
\begin{proof}
  By \cref{lem:ak-mean-variance,lem:vk-mean-variance}, we have $\Ex[V_k] =
  \Ex[a_k] = \delta_{k, 0}$, which in turn implies that
  \begin{equation*}
    \Var \! \left(\bigsum_{k=0}^\cut a_k z^k - \bigsum _{k=0}^\cut V_kz^k\right)
    = \Ex\! \abs{ \bigsum_{k=0}^\cut a_k z^k - \bigsum _{k=0}^\cut V_k z^k }^2.
  \end{equation*}
  Then, similar to \cref{lem:a_k-tail-bound}, we have
  \begin{align*}
    \Var \!\left( \sum_{k=0}^\cut a_k z^k - \sum _{k=0}^\cut V_k z^k \right)
     & = \sum_{k=0}^\cut \Ex \! \abs{ a_k-V_k }^2 \abs{z}^{2k}
     \le \sum_{k=0}^\cut \dfrac{k(k-1)}{2n \cdot k!} \abs{z}^{2k} \\
     & = \dfrac{\abs{z}^4}{2n} \sum_{k=0}^{\cut-2} \dfrac{\abs{z}^{2k}}{k!}
     \le \dfrac{\abs{z}^4 e^{\abs{z}^2}}{2n},
  \end{align*}
  where the first step is an analogue of \cref{lem:ak-mean-variance} useing the
  independence and zero mean value property of entries in $\bm{A}$
  and the last comes from the taylor expansion of $e^{\abs{z}^2}$.

  From Chebyshev's inequality, it follows that
  \begin{equation*}
    \Pr \left(\abs{\sum_{k = 0}^t a_k z^k - \sum_{k = 0}^\cut V_k z^k}
      \ge n^{-\beta} \right)
    \le \dfrac{\abs{z}^4 e^{\abs{z}^2}}{2\,n^{1 - 2\beta}}.
  \end{equation*}
  Since $\beta < \frac{1}{2}$ and $\abs{z} \le (\ln n)^{\frac{1}{8}}$, this
  probability is upper-bounded by $o(1)$.
\end{proof}

\section{Upper Bounds of the Power-Sum of Columns}

In this section, we establish the recursion for $V_k$'s and
a concentration bound of $D_k$. This uses the elementary symmetric
polynomial and moment inequalities.

\subsection{Newton's Identies in Terms of Symmetric Polynomials}

The elementary symmetric polynomials and power sums, as defined in
\cref{def:s-e-poly}, follow the so-called Newton's identities.
We give the following elementary derivation for reader's convenience.
\begin{lemma}
  \label{lem:product-sum-formula}
  Given variables $x_1, x_2, \ldots, x_n$ and any $m \in [n],$ we have
  \begin{equation*}
    e_m(n) = \dfrac{1}{m}\bigsum_{k=0}^{m-1}(-1)^k e_{m-k-1}(n)
    S_{k+1}(n).
  \end{equation*}
\end{lemma}

\begin{proof}
  Let us introduce auxiliary variables $Q_{m,k}$ defined by
  \begin{align*}
    Q_{m,k} \triangleq \bigsum_{\{j_1, \cdots, j_m\} \in C_{n, m}}
    x_{j_1}x_{j_2}\cdots x_{j_m}\bigsum_{i=1}^m x_{j_i} ^ {k-1}
  \end{align*}
  for any $0 \le m \le n, k \ge 1$ with the convention $Q_{0, k} \equiv 0$ for
  $k \ge 1.$ By definition, $Q_{m,1} = m e_m(n) , Q_{1,k} = S_k$.

  Then we consider a counting problem: choose a $(m-1)$-subset $A$ of $[n]$
  together with an $i \in [n]$, and the contribution of this choice is $x_i^k
  \prod_{j\in A} x_j$.
  On the other hand, we can partition all choices by the criterion whether $k\in
  A$.
  Thus
  \[
    Q_{m,k} + Q_{m - 1, k + 1} = e_{m - 1}(n) S_k
  \]
  holds for all $m \ge 1$.
  Solving $Q_{m,1}$, the lemma immediately follows.
\end{proof}

\subsection{A Third Moment Inequality}

By some calculation, we can derive the following upper-bound of the absolute
third moment of a sequence of \iid\ complex rv's.
\begin{lemma}
  \label{lem:sum-third-moment}
  Suppose $X_1, X_2, \cdots, X_n$ is a sequence of \iid\ random variables
  following distribution $\mathcal{D}$,
  then there exists an absolute constant $\eta > 0$ such that
  \begin{equation*}
    \Ex \abs{ \frac{\sum_{i=1}^n X_i}{\sqrt{n}} }^3 \le
    \eta \left(1 + \frac{\rho}{\sqrt{n}} \right).
  \end{equation*}
\end{lemma}

\begin{proof}
  Let
  \begin{equation*}
    \sigma_1 = \sqrt{\Ex_{X \sim \mathcal{D}}[\Re(X)^2]}, \quad
    \sigma_2 = \sqrt{\Ex_{X \sim \mathcal{D}}[\Im(X)^2]},
  \end{equation*}
  and
  \begin{equation*}
    x_i \triangleq \Re(X_{i}),\quad
    y_i \triangleq \Im(X_{i}),\quad
    \rho_1 \triangleq \Ex_{X \sim \mathcal{D}}\abs{\Re(X)}^3,\quad
    \rho_2 \triangleq \Ex_{X \sim \mathcal{D}}\abs{\Im(X)}^3.
  \end{equation*}
  Since $\rho < \infty$, $\rho_1$ and $\rho_2$ exists.
  Define $R_m \triangleq \sum_{j=1}^m x_j$ and $T_m \triangleq \sum_{j=1}^m
  y_j$ for $m \in [n].$

  For all integer $k$, we have
  \begin{equation*}
    \begin{split}
      \Ex \abs{R_{2k}}^3
      & = \Ex \abs{R_k + (R_{2k} - R_k)}^3\\
      & \le \Ex \left[ \abs{R_k} + \abs{R_{2k} - R_k} \right]^3\\
      & = \Ex \abs{R_k}^3 + \Ex \abs{R_{2k} - R_k}^3 + 3 \Ex \abs{R_k}^2
      \abs{R_{2k}-R_k} + 3 \Ex \abs{R_k} \abs{R_{2k} - R_k}^2\\
      & = 2 \Ex \abs{R_k}^3 + 6 \Ex R_k^2 \Ex \abs{R_k}\\
      & \le 2 \Ex \abs{R_k}^3 + 6 \Ex R_k^2 \sqrt{\Ex R_k^2}\\
      & = 2 \Ex \abs{R_k}^3 + 6 k \sigma_1^2 \cdot \sqrt{k} \sigma_1.
    \end{split}
  \end{equation*}
  Similarly,
  \begin{equation*}
    \begin{split}
      \Ex \abs{R_{2k+1}}^3
      & \le \Ex \left[ \abs{R_{2k}} + \abs{x_{2k+1}} \right]^3\\
      & \le \Ex \abs{R_{2k}}^3 + \Ex \abs{x_{2k+1}}^3 + 3 \Ex
      \abs{R_{2k}}^2 \abs{x_{2k+1}} + 3 \Ex \abs{R_{2k}} \abs{x_{2k+1}}^2\\
      & \le \Ex \abs{R_{2k}}^3 + \Ex \abs{x_{2k+1}}^3 + 3 \Ex
        R_{2k}^2 \sqrt{\Ex x_{2k+1}^2} + 3 \sqrt{\Ex
        R_{2k}^2} \Ex x_{2k+1}^2\\
      & = \Ex \abs{R_{2k}}^3 + \rho_1 + 6k\sigma_1^3 +
      3\sqrt{2k}\sigma_1^3.
    \end{split}
  \end{equation*}

  Applying induction using the above two rules, we have
  \begin{equation*}
    \begin{split}
      \Ex \abs{R_n}^3
      \le n \rho_1 + \sigma_1^3 \sum_{i \ge 1} \left[6 \left( \frac{n}{2^i}
      \right)^{3/2} + \frac{6 n}{2^i} + 3 \sqrt{2} \sqrt{
        \frac{n}{2^i} } \,\right]
      \le C' \left( n \rho_1 + n^{3/2} \sigma_1^3 \right)
    \end{split}
  \end{equation*}
  for some constant $C'$.
  A similar reasoning for the imaginary part gives
  \begin{equation*}
    \Ex \abs{T_n}^3 \le C' \left( n \rho_2 + n^{3/2} \sigma_2^3 \right).
  \end{equation*}

  For $0 \le k \le n$, we have
  \begin{equation*}
    \begin{split}
      \Ex \abs{R_n}^2 \abs{T_n}
      = & \Ex \abs{R_k + (R_n - R_k)}^2 \abs{T_k + (T_n - T_k)}\\
      \le & \Ex \Bigl[ \abs{R_k}^2 \abs{T_k} + \abs{R_k}^2 \abs{T_n - T_k} +
        2\abs{R_k} \abs{R_n - R_k} \abs{T_k}\\
        & \quad + 2\abs{R_k} \abs{R_n - R_k} \abs{T_n - T_k} + \abs{R_n - R_k}^2
        \abs{T_k} + \abs{R_n - R_k}^2 \abs{T_n - T_k} \Bigr]\\
      \le & \Ex R_k^2 \abs{T_k} + \Ex R_k^2 \sqrt{\Ex T_{n-k}^2} +
        2 \sqrt{\Ex R_k^2\, \Ex T_k^2} \sqrt{\Ex R_{n-k}^2}\\
      & \quad + 2 \sqrt{\Ex R_k^2}\sqrt{\Ex R_{n-k}^2\, \Ex T_{n-k}^2}
        + \Ex R_{n-k}^2 \sqrt{\Ex T_k^2}
        + \Ex R_{n-k}^2 \abs{T_{n-k}}\\
      = & \Ex\!\left[ \abs{R_k}^2 \abs{T_k} + \abs{R_{n-k}}^2 \abs{T_{n-k}}
        \right] + 3 \sigma_1^2 \sigma_2 \sqrt{k(n-k)}(\sqrt{k} + \sqrt{n-k}).
    \end{split}
  \end{equation*}
  This establishes that
  \begin{equation*}
    \Ex \abs{R_n}^2\abs{T_n} \le n \Ex \abs{x_1^2 y_1} + C'' n^{3/2} \sigma_1^2
    \sigma_2,
  \end{equation*}
  for some constant $C''$.
  Symmetrically, we have
  \begin{equation*}
    \Ex \abs{R_n}\abs{T_n}^2 \le n \Ex \abs{x_1 y_1^2} + C'' n^{3/2} \sigma_1
    \sigma_2^2.
  \end{equation*}
  Therefore,
  \begin{equation*}
    \begin{split}
      \Ex \abs{\frac{\sum_{j=1}^n X_j}{\sqrt{n}}}^3 =
      & n^{-\frac{3}{2}} \Ex \abs{R_n + i T_n}^3\\
      \le & n^{-\frac{3}{2}} \Ex \bigl[ \abs{R_n}^3 + \abs{T_n}^3 + 3
        \abs{R_n}^2\abs{T_n} + 3 \abs{R_n}\abs{T_n}^2 \bigr]\\
      \le & n^{-\frac{3}{2}} \Bigl[ C' n(\rho_1 + \rho_2) + C' n^{\frac{3}{2}}
        \bigl( \sigma_1^3 + \sigma_2^3 \bigr) + 3 C'' n^{\frac{3}{2}}
        \sigma_1\sigma_2 (\sigma_1 + \sigma_2)\\
      & \quad + 3n \Ex \bigl( \abs{x_1}^2\abs{y_1} + \abs{x_1}\abs{y_1}^2 \bigr)
        \Bigr]
    \end{split}
  \end{equation*}
  Using the basic inequality
  \begin{equation*}
    x^2 y + x y^2 \le x^3 + y^3, \quad \forall x, y > 0,
  \end{equation*}
  and the fact that
  \begin{equation*}
    \begin{cases}
      \rho_1, \rho_2 \le \rho, \\
      \sigma_1, \sigma_2 \le 1,
    \end{cases}
  \end{equation*}
  we have for some constant $\eta > 0$ that
  \begin{equation*}
    \begin{split}
      \Ex \abs{\frac{\sum_{j=1}^n X_j}{\sqrt{n}}}^3
      & \le n^{-\frac{3}{2}} \Bigl[ 2 C' n \rho + 2 C' n^{\frac{3}{2}} + 6 C''
      n^{\frac{3}{2}} + 3n \Ex \bigl( \abs{x_1}^3 + \abs{y_1}^3 \bigr) \Bigr] \\
      & \le \eta (1 + \frac{\rho}{\sqrt{n}}).
    \end{split}
  \end{equation*}
\end{proof}

\subsection{Bounds for $D_k$}

\begin{lemma} \label{lem:bound-d2}
  For any $0 < \phi < \frac{1}{2}$,
  \[
     \Pr\!\left(\vert D_2 - \xi \vert \le n^{-\phi} \right) = 1 - o(1) .
  \]
\end{lemma}
\begin{proof}
  Let
  \begin{equation*}
    X_{i, j} \triangleq a_{i, j} \mathbbm{1}_{\abs{a_{i,j}} \le n},\quad
    \mu_{k} \triangleq \Ex X_{i,j}^k,\quad
    \mu_k^* \triangleq \Ex \abs{ X_{i,j} }^k,\quad
    \mu^\dag \triangleq \Ex\! \left[ \abs{X_{i,j}}^2
    X_{i,j} \right].
  \end{equation*}
  Since all elements in $\bm{A}$ are \iid\ and $X_{i, j}$'s are bounded,
  these values are well-defined. Note that we only care about the asymptotic
  behaviour, we assume $n \ge \rho$ in the following proof.

  Observe that
  \begin{equation*}
    \Pr\! \left( \abs{ a_{i,j} } > n \right)
    \le \frac{\Ex \abs{ a_{i,j} }^3}{n^3}
    \le \frac{\rho}{n^3},
  \end{equation*}
  $\bm{A}$ satisfies
  \[
    \Pr \big(\exists i,j\in[n]: \abs{ a_{i,j} } > n \big)
    \le \frac{\rho}{n}.
  \]
  Therefore,
  \begin{equation}
    \label{eq:division-bound}
    \Pr\! \left( \abs{ D_2 - \xi } > \eps \right)
    \le \Pr\!\left( \abs{ \frac{\sum_{j=1}^n \left(\sum_{i=1}^n
    X_{i,j} \right)^2}{n^2} - \xi } > \eps \right) +
    \frac{\rho}{n}.
  \end{equation}

  Next, we bound some moments.
  \begin{itemize}
  \item $\abs{\mu_1} = \abs{- \Ex \left[ a_{i,j} \mathbbm{1}_{\{\abs{ a_{i,j}}>
          n\}} \right]} \le \Ex \left[ \abs{ a_{i,j} } \mathbbm{1}_{\vert
        a_{i,j} \vert > n} \right] \le \Ex \left[ \abs{ a_{i,j} } \left( \abs{
          a_{i,j} } / n \right)^2 \right]\le \dfrac{\rho}{n^2}$;
    \item $\abs{ \xi - \mu_2 } = \abs{ \Ex \left[ a_{1,1}^2 - a_{1,1}^2
          \mathbbm{1}_{\abs{ a_{1,1} }\le n} \right] } = \abs{ \Ex\! \left[
          a_{1,1}^2 \mathbbm{1}_{\abs{a_{1,1}}>n} \right]} \le \Ex\! \left[
        \abs{ a_{1,1} }^3/n \right] = \rho/n$, then $\abs{\mu_2} \le \abs{\xi} +
      \rho / n \le 1 + \rho / n \le 2$ since $n \ge \rho$;
    \item $\mu_2^* = \Ex\!\left[ \abs{ a_{i,j} }^2
      \mathbbm{1}_{\abs{ a_{i,j} } \le n} \right] \le \Ex\! \left[
      \abs{ a_{i,j} }^2 \right] = \Var\!\left[ a_{i,j} \right] = 1$;
    \item $\mu_4^* = \Ex\! \left[ \abs{ a_{i,j} }^4
      \mathbbm{1}_{\abs{ a_{i,j} }\le n} \right] \le n \Ex\!\left[
      \abs{a_{i,j}}^3 \right] \le n\rho$;
    \item $\abs{\mu^\dag} = \abs{\Ex\! \left[
          \abs{ a_{i,j} }^2 a_{i,j} \mathbbm{1}_{\abs{ a_{i,j} }
        \le n} \right] } \le \Ex \abs{ a_{i,j} }^3 = \rho$.
  \end{itemize}

  Let $S_{j} \triangleq \sum_{i=1}^n X_{i,j}, S \triangleq \sum_{j=1}^n (S_j^2 -
  n\xi )$.
  Since
  \begin{equation}
    \label{eq:sum-moment-bound}
    \begin{split}
      \Ex \abs{S}^2 =
      & \, \Var[S] + \abs{ \Ex[S] }^2 \\
      \le & \, n \Var\!\left[S_1^2 - n\xi \right] +
        n^2\abs{ \Ex\!\left[S_1^2-n\xi \right] }^2 \\
      =&\, n \Var\!\left[S_1^2 \right] + n^2 \abs{ \Ex\!\left[ S_1^2
        - n\xi \right] }^2\\
      \le&\, n \, \Ex \!\left[S_1^2\bar{S}_1^2 \right] + n^2\abs{
        \Ex\!\left[S_1^2 - n\xi \right] }^2.
    \end{split}
  \end{equation}
  We then bound $\abs{ \Ex\! \left[S_1^2 - n\xi \right] }$ and $\Ex
  \!\left[S_1^2\bar{S}_1^2 \right]$ separately.

  For the first part, since $n \ge \rho,$
  \begin{align*}
    \abs{ \Ex\!\left[ S_1^2 - n \xi \right] }
    = \abs{ n \mu_2 + n (n-1) \mu_1^2 - n \xi }
    \le n \abs{ \mu_2 - \xi } + n(n-1) \abs{ \mu_1 }^2
    \le \rho + 1.
  \end{align*}

  For the second, consider all five kinds of monomials in
  \[
    \Ex\! \left[S_1^2\,\overline{S}_1^2 \right]
    = \sum_{i,j,k,l\in[n]} \Ex\!\left[ X_{i,1} X_{j,1}
    \overline{X_{k,1}X_{l,1}} \right],
  \]
  we have
  \begin{align*}
    \Ex\! \left[ S_1^2\,\overline{S}_1^2 \right]
    =\, & \sum_{i=1}^n \Ex\! \left[ X_{i}^2 \overline{X}_{i}^2 \right]
          + \sum_{i\neq j} \Ex\!\left[
          X_{i}X_{j}\overline{X}_{j} \overline{X}_{j} +
          X_{j}X_{i}\overline{X}_{j} \overline{X}_{j} +
          X_{j}X_{j}\overline{X}_{i} \overline{X}_{j} +
          X_{j}X_{j}\overline{X}_{j} \overline{X}_{i} \right] \\
    \, & +\quad\, \sum_{i<j} \quad\, \Ex\!\left[
         X_{i}X_{i}\overline{X}_{j} \overline{X}_{j} +
         X_{j}X_{j}\overline{X}_{i} \overline{X}_{i} +
         4 X_{i}X_{j}\overline{X}_{i} \overline{X}_{j} \right] \\
    \, & + \sum_{(i, j, k) \in P_{n, 3},\, j < k} \Ex\! \left[
         2X_{i}X_{i}\overline{X}_{j} \overline{X}_{k} +
         2X_{j}X_{k}\overline{X}_{i} \overline{X}_{i} +
         4X_{i}X_{j}\overline{X}_{i} \overline{X}_{k} +
         4X_{i}X_{k}\overline{X}_{i} \overline{X}_{j} \right] \\
    \, & + \sum_{(i, j, k, l) \in P_{n, 4}} \Ex\! \left[
         X_i X_j \overline{X}_k \overline{X}_l \right]\\
    =\, & n \mu_{4}^* + n(n-1) \left(2 \mu_1 \overline{\mu^\dag} +
          2 \overline{\mu_1} \mu^\dag \right) + \binom{n}{2}
          \left[ 2\mu_2\overline{\mu_2} + 4(\mu_2^*)^2 \right] \\
    \, & + n\binom{n-1}{2} \left(2\mu_2 \overline{\mu_1}^2 + 2\mu_1^2
         \overline{\mu_2} + 8 \mu_2^* \mu_1 \overline{\mu_1} \right) +
         n^{\underline{4}} \mu_1^2 \overline{\mu_1^2} \\
    \le\, & n\mu_4^* + 2n^2 \left(\mu_1 \overline{\mu^\dag} +
            \overline{\mu_1} \mu^\dag \right) + n^2 \left[\mu_2
            \overline{\mu_2} + 2(\mu_{2}^*)^2 \right]\\
    \, & + n^3 \left(\mu_2 \overline{\mu_1}^2 + \mu_1^2 \overline{\mu_2} +
         4\mu_2^* \mu_1 \overline{\mu_1} \right) + n^4 \left\vert \mu_1
         \right\vert^4 \\
    \le\, & n^2\rho + 4n^2 \dfrac{\rho}{n^2} \cdot \rho +
      n^2 \left(2^2 + 2 \cdot 1^2 \right) + n^3 \left(2 \dfrac{\rho^2}{n^4} + 2
            \dfrac{\rho^2}{n^4} + 4 \cdot 1 \cdot \dfrac{\rho^2}{n^4} \right) +
            n^4 \frac{\rho^4}{n^8} \\
    \le\, & 20 n^2\rho,
  \end{align*}
  where we slightly abuse the notation to use $X_i$ to denote $X_{i,1}$ and
  use the assumption $n \ge \rho$ in the last step.

  Therefore, by \cref{eq:sum-moment-bound}, $n \ge \rho \ge
  (\sigma_1^2+\sigma_2^2)^{3/2} = 1$, and Chebyshev's inequality, we have
  \begin{equation*}
    \Pr\! \left( \abs{ \frac{S}{n^2} } >\eps \right)
    \le \frac{\Ex\! \left[ S\overline{S} \right]}{n^4 \eps^2}
    \le \frac{n \times 20n^2\rho + n^2(\rho + 1)^2}{n^4\eps^2}
    \le \frac{20n^3\rho + n^24\rho^2}{n^4\eps^2}
    \le \frac{24\rho}{n\eps^2}.
  \end{equation*}

  Taking $\eps = n^{-\phi}$ where $0 < \phi < \frac{1}{2}$, and applying
  \cref{eq:division-bound}, we have
  \begin{equation}
    \label{eq:tail-k=2}
    \Pr\!\left( \abs{D_2 - \xi} \le n^{-\phi} \right)
    \ge 1 - \dfrac{\rho}{n} - 24 \rho n^{2\phi - 1} = 1 - o(1).
  \end{equation}
\end{proof}

\begin{lemma}
  \label{lem:bound-dk}
  Fix any positive constant $\Delta < \frac{1}{6}$, it holds that
  \begin{equation*}
    \Pr\!\left(\forall\, k\ge 3,\, \abs{D_k} \le n^{-\Delta k}\right) = 1 - o(1).
  \end{equation*}
\end{lemma}

\begin{proof}
  The statement is equivalent to the following bound
  \begin{equation*}
    \Pr \Bigl( \exists\, k\ge 3,\, \abs{D_k} > n^{-\Delta k} \Bigr) = o(1).
  \end{equation*}
  The left hand side can be bounded as
  \begin{equation}
    \label{eq:bound-dk-1}
    \begin{split}
      \Pr \Bigl( \exists\, k\ge 3,\, \abs{D_k} > n^{-\Delta k} \Bigr)
      & = \Pr \Bigl( \exists\, k\ge 3,\, \frac{\sum_{j=1}^n
        \abs{C_j}^k}{n^{k/2}} > n^{-\Delta k} \Bigr)\\
      & = \Pr \biggl( \exists\, k\ge 3,\, \biggl( \sum_{j=1}^n \abs{C_j}^k
        \biggr)^{1/k} > n^{1/2-\Delta} \biggr)\\
      & \le \Pr \biggl( \biggl( \sum_{j=1}^n \abs{C_j}^3 \biggr)^{1/3} >
        n^{1/2-\Delta} \biggr),
    \end{split}
  \end{equation}
  where the last step follows from the well-known decreasing property of the
  $L^p$ norm,
  \begin{equation*}
    \biggl( \sum_{j=1}^n \abs{C_j}^k \biggr)^{1/k} \le
    \biggl( \sum_{j=1}^n \abs{C_j}^3 \biggr)^{1/3}\quad \forall\, k \ge 3.
  \end{equation*}

  Recall that by \cref{lem:sum-third-moment}, there is a constant $\eta > 0$,
  such that
  \begin{equation*}
    \Ex\! \abs{C_j}^3 \le \eta \!\left(1 + \frac{\rho}{\sqrt{n}} \right)
    \quad \forall\, j \in [n].
  \end{equation*}
  We can continue the bound by Markov's inequality in \cref{eq:bound-dk-1} as
  \begin{equation*}
    \Pr \Bigl( \exists\, k\ge 3,\, \abs{D_k} > n^{-\Delta k} \Bigr)
    \le \frac{n\, \Ex \abs{C_j}^3}{n^{3/2-3\Delta}}
    \le \frac{\eta (1 + \rho n^{-1/2})}{n^{1/2-3\Delta}}.
  \end{equation*}
  The right hand side is $o(1)$ for $\Delta< 1/6$ and this proves the lemma.
\end{proof}

\section{Explicit Expression and Upper-bounds of $V_k'$}

In this section, we solve the recursion of $V_k'$ utilizing the well-known
``probabilists' Hermite polynomials'' and establish some bounds of $V_k'$ which
will be used to bound the difference between $V_k$ and $V_k'$ in the next
section.

\subsection{Probabilists' Hermite Polynomials}

The ``probabilists' Hermite polynomials'' are given by
\begin{equation*}
  H_{e_n}(x)
  = (-1)^n e^{\frac{x^2}{2}} \frac{\diff^n}{\diff x^n} e^{-\frac{x^2}{2}}
  = \left(x - \frac{\diff}{\diff x} \right)^n \cdot 1
\end{equation*}
for $n \in \N$.
Solving the equation, we can get its explicit expression
\begin{equation}
  \label{eq:hermite-poly-definition}
  H_{e_n}(x) = n! \sum_{k=0}^{\lf \frac{n}{2} \rf}
  \frac{(-1)^k x^{n-2k}}{k!\, (n - 2k)!\, 2^k},
\end{equation}
a polynomial of degree $n$.
Define
\begin{equation}
  h_n(x) \triangleq \frac{1}{n!} H_{e_n}(x).
\end{equation}

Note that $H_{e_n}(x)$ satisfies
\begin{equation*}
  H_{e_n}(x) =
  \begin{cases}
    1 , & n = 0, \\
    x , & n = 1, \\
    x H_{e_{n - 1}}(x) - (n - 1) H_{e_{n - 2}}(x) , & n \ge 2.
  \end{cases}
\end{equation*}
We can derive a similar recursion for $h_n(x)$
\begin{equation} \label{eq:f-n-recursion}
  h_n(x) =
  \begin{cases}
    1 , & n = 0, \\
    x , & n = 1, \\
    \dfrac{x h_{n - 1}(x) - h_{n - 2}(x)}{n} , & n \ge 2.
  \end{cases}
\end{equation}

The following upper bound on the Hermite polynomials will be useful in later
proofs.

\begin{lemma}
  \label{lem:hermite-bound}
  For any $n \in \mathbb{N}$ and any $x \in \mathbb{C}$, it holds that
  \begin{equation*}
    \abs{h_n(x)} \le \max(1, \abs{x})^n \left(
      \frac{n}{e^2}\right)^{-\frac{n}{2}}.
  \end{equation*}
\end{lemma}

\begin{proof}
  By the definition of $h_n(x)$, we have
  \begin{equation*}
    \begin{split}
      \abs{h_n(x)} \le \, & \sum_{k=0}^{\lf \frac{n}{2} \rf}
      \frac{\abs{x}^{n-2k}}{k!\, (n-2k)!\, 2^k}\\
      \le \, & \sum_{k=0}^{\lf \frac{n}{2} \rf} \abs{x}^{n-2k} \left(\frac{2k}{e}
      \right)^{-k} \left(\frac{n - 2k}{e} \right)^{-n+2k}.
    \end{split}
  \end{equation*}
  Use $\phi(k)$ to denote the inverse of the coefficient of the $k$-th term of
  the previous equation,
  \begin{equation*}
    \phi(k) \triangleq \left(\frac{2k}{e} \right)^{k} \left(\frac{n - 2k}{e}
    \right)^{n - 2k} > 0
  \end{equation*}
  for $k \in \bigl[0, \lf \frac{n}{2} \rf \bigr]$.
  The derivative of $\ln \phi(k)$ is
  \begin{equation*}
    \frac{\diff}{\diff k} \ln \phi(k) = \ln (2k) - 2 \ln (n - 2k),
  \end{equation*}
  showing that the minimum value of $\phi(k)$ is achieved at
  \begin{equation*}
    k_0 = \frac{n}{2} + \frac{1}{4} - \frac{\sqrt{4n+1}}{4}.
  \end{equation*}
  For $n \ge 1$, we have
  \begin{equation*}
    \begin{cases}
      k_0 \ge \dfrac{n - \sqrt{n}}{2} \ge 0,\\
      n - 2k_0 \ge \sqrt{n} - \dfrac{1}{2} \ge 0,
    \end{cases}
  \end{equation*}
  and furthermore
  \begin{align*}
    \phi(k) \ge \phi(k_0)
    \ge\, & \left(\frac{n - \sqrt{n}}{e} \right)^{\frac{n - \sqrt{n}}{2}}
            \left(\frac{\sqrt{n}-\frac{1}{2}}{e}\right)^{\sqrt{n}-\frac{1}{2}} \\
    =\, & e^{-\frac{n + \sqrt{n} - 1}{2}} \left(n - \sqrt{n} \right)^{\frac{n
          - \sqrt{n}}{2}} \left(n - \sqrt{n} + \frac{1}{4} \right)^{\frac{
          \sqrt{n}}{2} - \frac{1}{4}} \\
    >\, & e^{-\frac{n + \sqrt{n} - 1}{2}} \left(n - \sqrt{n}
          \right)^{\frac{n}{2} - \frac{1}{4}} \\
    =\, & e^{-\frac{n + \sqrt{n} - 1}{2}} n^{\frac{n}{2}}
       \left(n - \sqrt{n}\right)^{-\frac{1}{4}}
       \left(1 - n^{-\frac{1}{2}} \right)^{\frac{n}{2}} \\
    = & \left(\frac{n}{e} \right)^{\frac{n}{2}} e^{-\frac{\sqrt{n}-1}{2}}
        \left(n - \sqrt{n} \right) ^ {-\frac{1}{4}}
        \left(1 - n^{-\frac{1}{2}} \right)^{\frac{n}{2}}.
  \end{align*}
  Since $\left(1 - \frac{1}{x} \right)^x$ is increasing in $(1, \infty)$,
  it holds that for all $n \ge 4$,
  \begin{equation*}
    \left( 1 - n^{-\frac{1}{2}} \right)^{\frac{n}{2}}
    = \left[\left(1 - n^{-\frac{1}{2}}\right)^{\sqrt{n}}
    \right]^{\frac{\sqrt{n}}{2}}
    \ge \left( 1 - \frac{1}{\sqrt{4}} \right)^{\sqrt{4} \cdot
      \frac{\sqrt{n}}{2}}
    = 2^{-\sqrt{n}}
  \end{equation*}
  and then we can continue the bound on $\phi(k)$ as
  \begin{align*}
    \phi(k) \ge & \left( \frac{n}{e} \right)^{\frac{n}{2}}
                  e^{-\frac{\sqrt{n}-1}{2}} \times n^{-\frac{1}{4}}
                  \times 2^{-\sqrt{n}}\\
    = & \left(\dfrac{n}{e} \right)^{\frac{n}{2}}
        \exp\!\left(-\frac{\sqrt{n}-1}{2} - \frac{1}{4} \ln n
        - \sqrt{n} \ln 2 \right).
  \end{align*}
  When $n \ge 25$, it holds that
  \begin{equation*}
    -\frac{\sqrt{n}-1}{2} - \frac{1}{4} \ln n
        - \sqrt{n} \ln 2 \ge -\frac{n}{2} + \ln n,
  \end{equation*}
  which implies
  \begin{align}
    \label{eq:phi-k-bound}
    \phi(k) \ge \left(\frac{n}{e}\right)^{\frac{n}{2}}
    \exp\!\left(-\frac{n}{2} + \ln n \right)
    = n \left( \dfrac{n}{e^2} \right)^{\frac{n}{2}}
    \ge \dfrac{n + 2}{2} \left( \dfrac{n}{e^2} \right)^{\frac{n}{2}}.
  \end{align}
  Checking the remaining cases by hand, we conclude that
  \begin{equation*}
    \phi(k) \ge \dfrac{n + 2}{2} \left(\dfrac{n}{e^2} \right)^{\frac{n}{2}}
  \end{equation*}
  holds for $n \in \mathbb{N}, 0 \le k \le \lf \frac{n}{2} \rf$ with convention
  that $0^0 = 1$.
  Thus,
  \begin{equation*}
    \begin{split}
      \abs{h_n(x)} \le & \, \sum_{k=0}^{\lf \frac{n}{2} \rf} \abs{x}^{n-2k} / \phi(k)\\
      \le & \, \frac{n+2}{2} \cdot \max(1, \abs{x})^{n} \cdot
      \frac{2}{n+2} \left( \frac{n}{e^2} \right)^{-\frac{n}{2}}\\
      = & \, \max(1, \abs{x}^n) \left( \frac{n}{e^2} \right)^{-\frac{n}{2}}.
      \qedhere
    \end{split}
  \end{equation*}
\end{proof}

\subsection{Upper-bound of $V_k'$}

Comparing the recursion of $V_k'$ in \cref{eq:v-k'-recursion} to that of
$h_k(x)$ from \cref{eq:f-n-recursion}, we have
\begin{equation}
  \label{eq:vc'-formula}
  \begin{cases}
    V_k' = \dfrac{V_1^k}{k!} & \text{ if } \xi = 0, \\
    V_k' = \xi^{\frac{k}{2}} h_k\!\left(\frac{V_1}{\sqrt{\xi}} \right)\ &
    \text{ otherwise}.
  \end{cases}
\end{equation}

\Cref{lem:hermite-bound} can be used to establish an upper bound of $V_k'$ by
using \cref{eq:vc'-formula}.

\begin{lemma}
  \label{lem:V1}
  For all function $\theta(n) = \omega(1)$, it holds that
  \begin{equation*}
    \Pr\! \left( \abs{V_1} \le \theta \right) = 1 - o(1).
  \end{equation*}
\end{lemma}

\begin{proof}
  By the Chebyshev's inequality, we have
  \begin{equation*}
    \Pr\! \left( \abs{V_1} > \theta \right) \le \frac{\Var(V_1)}{\theta^2} =
    o(1).
    \qedhere
  \end{equation*}
\end{proof}

\begin{lemma}
  \label{lem:Vk'-V1}
  For any $k \in \N$, it holds that
  \begin{equation*}
    \abs{V_k'} \le \max\! \left( 1, \abs{V_1}^k \right)
      \left( \frac{k}{e^2} \right)^{-\frac{k}{2}}.
  \end{equation*}
  Note that $k$ might be larger than $n$ for notation convenience in
  \cref{eq:v-k'-summation}.
\end{lemma}

\begin{proof}
  Consider the following two cases depending on whether $\xi = 0$ or not.
  \begin{enumerate}
  \item $\xi = 0$.
    By definition, we have $V'_k = \dfrac{V_1^k}{k!}$ and
    \begin{equation*}
      \abs{V_k'} = \frac{\abs{V_1}^k}{k!}
      \le \abs{V_1}^k \left( \frac{k}{e} \right)^{-k}
      \le \max\! \left( 1, \abs{V_1}^k \right)
      \left( \frac{k}{e^2} \right)^{-\frac{k}{2}}
    \end{equation*}
    for any $k \ge 0$.
  \item $\xi \neq 0$.
    Recall that $V_k' = \xi^{\frac{k}{2}} h_k \!\left(\frac{V_1'}{\sqrt{\xi}}
    \right)$.
    We can apply \cref{lem:hermite-bound} as follows.
    \begin{align*}
      \abs{V_k'} & = \abs{\xi}^{\frac{k}{2}} \abs{h_k \left(
                   \frac{V_1}{\sqrt{\xi}} \right)}\\
                 & \le \abs{\xi}^{\frac{k}{2}} \max\! \left(
                   1, \abs{\frac{V_1}{\sqrt{\xi}}} \right)^k
                   \left( \frac{k}{e^2} \right)^{-\frac{k}{2}}\\
                 & \le \max\! \left(1, \abs{V_1}^k \right)
                   \left( \frac{k}{e^2} \right)^{-\frac{k}{2}},
    \end{align*}
    where in the final step we used the fact that $\abs{\xi} \le 1$.
    \qedhere
  \end{enumerate}
\end{proof}

\begin{lemma}
  \label{lem:Vk'}
  Let $\theta \triangleq \theta(n) = o(\sqrt[4]{\ln n})$ be a function of $n$
  such that $\theta \ge 1$ and $\abs{ V_1 } \le \theta$.
  Fixing any constant $\tau > 0$, for sufficiently large $n$ and any $k \in \N$,
  it holds that
  \begin{equation*}
    \abs{V_k'} \le n^\tau k^{-\frac{k}{4}}.
  \end{equation*}
  Additionally, we have the uniform bound
  \begin{equation*}
    \abs{V_k'} \le e^{\, 2 \, \theta^2}.
  \end{equation*}
\end{lemma}

\begin{proof}
  By \cref{lem:Vk'-V1}, we have
  \begin{equation*}
    \abs{V_k'} \le \max \left( 1, \abs{V_1}^{k} \right)
    \left( \frac{k}{e^2} \right)^{-\frac{k}{2}}.
  \end{equation*}
  This together with $\abs{V_1} \le \theta$ implies that, for all $k \ge 0$,
  \begin{equation*}
    \abs{ V_k' } \le \theta^{k} \, e^k\, k^{-\frac{k}{2}}
    = \exp\! \left( k \ln \theta + k - \frac{k \ln k}{2} \right) =
    (*)
  \end{equation*}
  Define function
  \begin{equation*}
    \phi(x) = \theta^x\, e^x\, x^{-\frac{x}{4}},
  \end{equation*}
  for $x \ge 0$.
  Calculating the derivative of $\ln \phi(x)$, we see that the maximum value of
  $\phi(x)$ is achieved at $x = e^3\, \theta^4$ and
  \begin{equation*}
    \phi(x) \le \phi(e^3\, \theta^4)
    = \exp \left( \frac{e^3\, \theta^4}{4} \right) = n^{o(1)},
  \end{equation*}
  where in the last step we use the condition $\theta = o(\sqrt[4]{\ln n})$.
  Then for sufficiently large $n$, $\phi(x)$ is bounded by $n^{\tau}$, which
  means
  \begin{equation*}
    \abs{V_k'} \le \phi(k)\, k^{-\frac{k}{4}} \le n^\tau k^{-\frac{k}{4}}.
  \end{equation*}

  For the uniform bound, by calculating the derivative of
  $(*)$ it follows that
  \begin{equation}
    \nonumber
    \abs{V_k'} \le \exp\!\left(k \ln \theta +
    k - \frac{k \ln k}{2} \right) \Bigg|_{k = e\, \theta^2} = \exp \left(
    \frac{e\, \theta^2}{2} \right) <  e^{\, 2\, \theta^2}.
  \end{equation}
  \qedhere
\end{proof}

\subsection{Summation of $V_k'$}
In view of the following two well-known expansion formula, for any $z, t \in
\C$,
\begin{equation}
  \nonumber
  \begin{cases}
    \displaystyle \sum_{k = 0}^\infty \frac{z^k}{k!} = e^z, \\
    \displaystyle \sum_{k = 0}^\infty \frac{H_{e_k}(z) t^k}{k!}
    = e^{zt -\frac{t^2}{2}}.
  \end{cases}
\end{equation}
Thus, \cref{eq:vc'-formula} implies
\begin{equation}
  \label{eq:v-k'-summation}
  \sum_{k = 0}^\infty V_k' z^k =
  \begin{cases}
    \begin{rcases}
      \displaystyle \sum_{k = 0}^\infty \frac{V_1^k z^k}{k!} = e^{V_1 z}, &
      \xi = 0 \\
      \displaystyle \sum_{k = 0}^\infty \frac{\sqrt{\xi}^k H_{e_k}
        \left(\frac{V_1}{\sqrt{\xi}} \right) z^k}{k!} =
      e^{V_1 z - \frac{\xi z^2}{2}}, & \xi \neq 0
    \end{rcases}
    = e^{V_1 z - \frac{\xi z^2}{2}}.
  \end{cases}
\end{equation}

With help of \cref{lem:Vk'}, we prove the following tail bound.
\begin{lemma}
  \label{lem:v-k-summation-tail}
  With all parameters satisfying \cref{main-parameters}, w.p. $1 - o(1)$,
  \begin{equation*}
    \abs{ \sum_{k = \cut + 1}^{\infty} V_k'\, z^k } =
    n^{-\omega(1)}.
  \end{equation*}
\end{lemma}

\begin{proof}
  Applying \cref{lem:V1} with $\theta(n) = \ln\ln n$ as in
  \cref{main-parameters}, w.p. $1 - o(1)$,
  \[
    \abs{V_1} \le \theta.
  \]
  In this case, it follows from \cref{lem:Vk'} that
  \begin{equation}
    \nonumber
    \abs{\sum_{k = \cut +1}^\infty V_k' z^k} \le  \sum_{k=\cut + 1}^\infty
    \abs{V_k'} \abs{z}^k
    \le  n^\tau \sum_{k=\cut + 1}^\infty k^{-\frac{k}{4}} \abs{z}^k.
  \end{equation}
  As in \cref{main-parameters}, $\abs{z}^8 \le \ln n < t$,
  which means for sufficiently large $n$,
  \begin{equation*}
    \frac{(k+1)^{-\frac{k+1}{4}}\abs{z}^{k+1}}{k^{-\frac{k}{4}}\abs{z}^k}
    = \frac{\abs{z}}{\sqrt[4]{k+1}}\left( 1 + \frac{1}{k} \right)^{-\frac{k}{4}}
    < \frac{\abs{z}}{\sqrt[4]{\cut}} \le \frac{1}{2}.
  \end{equation*}
  Thus,
  \begin{equation*}
    \abs{ \sum_{k=\cut+1}^\infty V_k' z^k} \le n^\tau t^{-\frac{t}{4}}\abs{z}^t
    = n^{-\omega(1)}.
  \end{equation*}
\end{proof}

Since $\vert e^z \vert = e^{\Re(z)}$ holds for $z \in \C$.
\cref{eq:v-k'-summation} says that the summation is small only if
the $\Re(V_1')$ is
small, which has small probability by concentration of $V_1$. Formally,
\begin{lemma}
  \label{lem:v'-summation-big}
  With all parameters satisfying \cref{main-parameters},
  \begin{equation*}
    \Pr\! \left[ \abs{ e^{V_1 z - \frac{\xi z^2}{2}} } \ge n^{-\gamma} \right]
    = 1 - o(1).
  \end{equation*}
\end{lemma}

\begin{proof}
  We upper bound the probability
  \begin{equation}
    \label{eq:V'-sum-1}
    \begin{split}
      \Pr\! \left[\abs{ e^{V_1 z - \frac{\xi z^2}{2}} } <
        n^{-\gamma} \right]
      = & \Pr\! \left[ \Re\! \left(V_1 z - \frac{\xi z^2}{2} \right) <
        -\gamma \ln n \right] \\
      = & \Pr\! \left[\Re\! \left(V_1 \frac{z}{\abs{z}} \right) <
        -\frac{\gamma \ln n}{\abs{z}} + \frac{\Re(\xi z^2) }{2 \abs{z}}
      \right].\\
    \end{split}
  \end{equation}
  Since $\abs{z}^8 \le \ln n$, for large $n$, it holds that
  \begin{equation*}
    \frac{\Re(\xi z^2) }{2 \abs{z}} \le \abs{\frac{\xi z^2}{2 z}} \le
    \frac{\gamma \ln n}{2 \abs{z}}.
  \end{equation*}
  Therefore, we can continue \cref{eq:V'-sum-1} as
  \begin{equation*}
    \Pr\! \left[\abs{ e^{V_1 z - \frac{\xi z^2}{2}} } <
        n^{-\gamma} \right] \le \Pr\! \left[ \abs{V_1} > \frac{\gamma \ln
          n}{2 \abs{z}}\right],
  \end{equation*}
  which is easily shown to be $o(1)$ applying \cref{lem:V1}.
\end{proof}

\section{Difference of $V_k$ and $V_k'$}
In this section, we bound the difference between $V_k$ and $V_k'$.
To this end, we simply apply triangle inequality of absolute values and
induction repeatedly.

\begin{lemma}
  \label{lem:eps-bound}
  With $\theta(n) = \ln\ln n$ as in \cref{main-parameters},
  fixing any positive constant $\nu < \frac{1}{8}$, there exists a constant $n_k
  = n_k(\sigma_1, \sigma_2, \delta, \rho)$ such that for any $n \ge n_k$, w.p.
  $1 - o(1)$, the difference $\eps_k \triangleq \abs{V_k' - V_k}$ is bounded by
  \begin{equation*}
    \eps_k \le n^{-\nu} k^{-\nu k}
  \end{equation*}
  for any $0 \le k \le n$.
\end{lemma}

\begin{proof}
  Recall that $V_0 = V_0' \equiv 1$ and $V_1 \equiv V_1'$ by definition.
  This gives $\eps_0 = \eps_1 = 0$.

  For $k \ge 2$, the triangle inequality and the bound $\vert \xi \vert \le 1$
  as proved in \cref{eq:xi-bound} establish the following upper bound
  \begin{equation*}
    \begin{split}
      k\, \eps_k = & \abs{\bigl( V_{k-1}' V_1' - V_{k-2}' \xi \bigr) -
       \Bigl( V_{k-1}V_1 - V_{k-2}D_2 + \sum_{i=2}^{k-1} (-1)^i V_{k-1-i}
       D_{i+1} \Bigr)}\\
     \le & \abs{ V_{k-1}'V_1 - V_{k-1}V_1 - V_{k-2}'\xi +
       V_{k-2}\xi - V_{k-2}\xi + V_{k-2} D_2} + \sum_{i=2}^{k-1}
     \abs{V_{k-1-i} D_{i+1}}\\
     \le & \abs{V_1} \eps_{k-1} + \eps_{k-2} + \abs{V_{k-2}}
     \abs{D_2 - \xi} + \sum_{i=2}^{k-1} \abs{ V_{k-1-i} } \abs{
       D_{i+1}}.\\
   \end{split}
 \end{equation*}
 Therefore, we can bound $\eps_k$ as
 \begin{equation}
   \label{eq:eps-bound-1}
   \begin{split}
     \eps_k \le & \frac{1}{k} \Bigl(
     \abs{V_1} \eps_{k-1} + \eps_{k-2} + \eps_{k-2} \abs{D_2 - \xi} +
     \abs{V_{k-2}'} \abs{D_2 - \xi} +\\
     & \qquad \sum_{i=2}^{k-1} \abs{ V_{k-1-i}' } \abs{D_{i+1}} +
     \sum_{i=2}^{k-1} \eps_{k-1-i} \abs{D_{i+1}} \Bigr).\\
    \end{split}
  \end{equation}

  Choose $\tau$ and $\Delta$ such that
  \begin{equation}
    \label{parameters-in-difference-bound}
    \begin{cases}
      \tau > 0, \\
      \Delta > \nu, \\
      \frac{1}{8} < \Delta < \frac{1}{6}, \\
      2 \tau + \nu < 2 \Delta.
    \end{cases}
  \end{equation}
  Applying \cref{lem:V1} with $\theta(n) = \ln\ln n$ as assumed,
  \begin{equation}
    \nonumber
    \Pr \left( \abs{V_1} \le \theta \right) = 1 - o(1).
  \end{equation}
  Plus \cref{lem:bound-d2,lem:bound-dk}, w.p. $1 - o(1)$, it holds that
  \begin{equation}
    \label{eq:eps-bound-d}
    \begin{cases}
      \abs{V_1'} \le \theta,\\
      \abs{D_2 - \xi} \le n^{-2\Delta},\\
      \abs{D_k} \le n^{-k\Delta}.
    \end{cases}
  \end{equation}
  For the rest of this proof, we assume that \cref{eq:eps-bound-d} holds.
  In this case,
  \begin{equation}
    \label{eq:eps-bound}
    \begin{split}
      \eps_k & \le \frac{1}{k} \Bigl[ \abs{V_1'} \eps_{k-1} + \eps_{k-2} +
      \sum_{i=2}^k \bigl( \eps_{k-i} + \abs{V_{k-i}'} \bigr)
      n^{-\Delta i} \Bigr]\\
      & \le \frac{1}{k} \Bigl[ \theta (\eps_{k-1} + \eps_{k-2}) +
      \sum_{i=0}^{k-2} \bigl( \eps_{i} + \abs{V_{i}'} \bigr) n^{-\Delta (k-i) }
      \Bigr]\\
    \end{split}
  \end{equation}

  We prove the claim by considering two cases $k \le \frac{\ln n}{\ln\ln n}$
  and $k > \frac{\ln n}{\ln\ln n}$.

  We first apply induction for $k \le \frac{\ln n}{\ln \ln n}$. The base cases
  for $k = 0, 1$ holds simply because $\eps_0 = \eps_1 \equiv 0$.
  Assume $\eps_j < n^{-\nu} j^{-\nu j} < 1$ holds for any $j < k$,
  by the uniform upper bound on $V_k'$ proven in \cref{lem:Vk'}
  and \cref{eq:eps-bound}, we have
  \begin{equation*}
    \begin{split}
      \eps_k \le & \frac{1}{k} \Bigl[ \theta (\eps_{k-1} + \eps_{k-2}) +
      \sum_{i=2}^k \bigl( \eps_{k-i} + e^{\, 2\theta^2} \bigr) n^{-\Delta i}
      \Bigr]\\
      \le & \frac{1}{k} \Bigl[ \theta (\eps_{k-1} + \eps_{k-2}) +
      2 \sum_{i=2}^k e^{\, 2\theta^2} n^{-\Delta i}
      \Bigr]\\
      \le & \theta (\eps_{k-1} + \eps_{k-2}) +
      2\, e^{\, 2\theta^2} n^{-2\Delta}.
    \end{split}
  \end{equation*}
  Define $\theta' = 3\, e^{\, 2\theta^2} n^{-2\Delta}$.
  The above equation can be relaxed as
  \begin{equation*}
    \eps_k \le \theta(\eps_{k-1} + \eps_{k-2}) + \theta'.
  \end{equation*}
  Using an induction on $k$, it is easy to see that $\eps_k \le \theta^k \theta'
  3^k$ since $\theta(n) > 1$ for large $n$.
  That is, for large $n$,
  \begin{equation*}
    \eps_k
    \le 3^{k+1} \theta^k e^{\, 2\theta^2} n^{-2\Delta}
    < n^{-\nu} k^{-\nu k}
  \end{equation*}
  holds since $\nu < \Delta$ as in \cref{parameters-in-difference-bound}.

  Now consider the case when $k \ge \frac{\ln n}{\ln \ln n}$
  and we will prove by another induction on $k$.
  The base case for $k = \frac{\ln n}{\ln \ln n}$ is proven in the previous
  case.
  Assume $\eps_j \le n^{-\nu} j^{-\nu j}$ holds for all $j<k$,
  we then prove the bound for $j=k$ as follows.

  First, we bound the summation $A \triangleq \sum_{j=0}^{k-2} \abs{V_{j}'}
  n^{-\Delta (k-j)}$ as
  \begin{equation*}
    A \le \, n^\tau \sum_{j=0}^{k-2} j^{-\frac{j}{4}} n^{-\Delta(k-j)}
    \le \, n^{\tau - \Delta k} \sum_{j=0}^{k-2} j^{-\frac{j}{4}}n^{\Delta j}.
  \end{equation*}
  Define function $\psi(x) = x^{-\frac{x}{4}} n^{\Delta x}$.
  The above equation can be written as
  \begin{equation}
    \label{eq:eps-bound-2}
    A \le n^{\tau - \Delta k} \sum_{j=0}^{k-2} \psi(j).
  \end{equation}
  Computing the derivative of $\ln \psi(x)$, it is easy to see that $\psi(x)$
  is increasing for $x \in \bigl[ 0, \frac{n^{4\Delta}}{e} \bigr]$ and
  decreasing for $x \in \bigl[ \frac{n^{4\Delta}}{e}, \infty \bigr)$.
  If $k-2 \le \frac{n^{4\Delta}}{e}$, \cref{eq:eps-bound-2} can be bounded as
  \begin{equation*}
    \begin{split}
      A \le & \, n^{\tau - \Delta k} \, (k-1) \, \psi(k-2)\\
      \le & \, n^{\tau - k \Delta} \, (k-1) (k-2)^{-\frac{k-2}{4}}n^{(k-2)\Delta}\\
      \le & \, n^{\tau - 2 \Delta} \, k^{-\frac{k}{8}} \\
      \le & \, n^{2\tau - 2 \Delta} \, k^{-\frac{k}{8}},
    \end{split}
  \end{equation*}
  where the third step holds for sufficiently large $n$
  since $k \ge \frac{\ln n}{\ln\ln n}$.
  Similarly, if $k-2 > \frac{n^{4\Delta}}{e}$, \cref{eq:eps-bound-2} can be
  bounded as
  \begin{equation*}
    \begin{split}
      A \le & \, n^{\tau - \Delta k} \, (k-1) \, \psi
      \left( \frac{n^{4\Delta}}{e} \right)\\
      = & \, n^{\tau - k \Delta} \, (k-1) \, e^{\frac{n^{4\Delta}}{4e}}\\
      \le & \, n^{\tau - 2 \Delta}\, n^{-(k-2)\Delta}\, (k-1) \,
      e^{\frac{k-2}{4}} \\
      \le & \, n^{\tau - 2 \Delta}\, k^{-(k-2)\Delta}\, (k-1) \,
      e^{\frac{k-2}{4}}\\
      \le & n^{2\tau - 2\Delta} \, k^{-\frac{k}{8}},
    \end{split}
  \end{equation*}
  where the third step holds for the condition $k-2 > \frac{n^{4\Delta}}{e}$,
  and the last step holds for sufficiently large $n$
  since $\Delta > \frac{1}{8}$ as in
  \cref{parameters-in-difference-bound}.
  Combining the above two cases, the following inequality holds
  \begin{equation*}
    A \le n^{2\tau - 2\Delta} k^{-\frac{k}{8}}.
  \end{equation*}

  \Cref{eq:eps-bound} then implies that
  \begin{equation}
    \eps_k \le \frac{1}{k} \Bigl[ \theta (\eps_{k-1} + \eps_{k-2}) +
    \sum_{j=0}^{k-2} \eps_{j} n^{-\Delta (k-j)} + n^{2\tau - 2\Delta}
    k^{-\frac{k}{8}}\Bigr].
  \end{equation}
  We bound each term of the summation as follows.
  \begin{itemize}
  \item First,
    \begin{equation}
      \label{eq:eps-bound-case-1}
      \begin{split}
        \theta (\eps_{k-1} + \eps_{k-2}) & \le 2 \theta\, n^{-\nu} (k-2)^{-\nu
          (k-2)}\\
        & = \frac{1}{2}{k\, n^{-\nu} k^{-\nu k}} \cdot 4 \theta k^{2\nu-1} \left(
          1 + \frac{2}{k-2} \right)^{\nu (k-2)}\\
        & \le \frac{1}{2}{k\, n^{-\nu} k^{-\nu k}} \cdot 4 \theta k^{2\nu-1}
        e^{2\nu}\\
        & \le \frac{1}{2}{k\, n^{-\nu} k^{-\nu k}},
      \end{split}
    \end{equation}
    where the final step holds for sufficiently large $n$ since
    $\nu < \frac{1}{8} < \frac{1}{2}$ as in assumption.
  \item Second, since
    \begin{equation*}
      \frac{n^{-\nu - \Delta(k - j)} j^{-\nu j}}
        {n^{-\nu - \Delta(k - j - 1)} (j + 1)^{-\nu (j + 1)}}
      = (j+1)^\nu \frac{\left( 1 + \frac{1}{j} \right)^{\nu j}}{n^\Delta}
      \le \frac{(j+1)^\nu e^\nu}{n^\Delta} <
      \frac{1}{2},
    \end{equation*}
  for sufficiently large $n$ and the choice $\nu < \Delta$
  in \cref{parameters-in-difference-bound}, we have
  \begin{equation*}
    \bigsum_{j=0}^{k-2} \eps_j n^{-\Delta (k-j)}
    \le \bigsum_{j = 0}^{k - 2} n^{-\nu - \Delta(k - j)} j^{-\nu j}
    \le 2 n^{-\nu - 2 \Delta} (k - 2)^{-\nu(k - 2)}.
  \end{equation*}
  For sufficiently large $n$, we have
  \begin{equation*}
    \sum_{j=0}^{k-2} \eps_j n^{-\Delta (k-j)}
    \le \frac{1}{4} k\, n^{-\nu} k^{-\nu k}.
  \end{equation*}
  Here we use the fact that $\nu, \Delta > 0$.
  \item Lastly, for large $n$,
    \begin{equation*}
      n^{2\tau-2\Delta} k^{-\frac{k}{8}} \le \frac{1}{4} k\, n^{-\nu} k^{-\nu k}
    \end{equation*}
    holds because $2\Delta > 2\tau + \nu$ in
    \cref{parameters-in-difference-bound}
    and $\nu < \frac{1}{8}$ in the statement.
  \end{itemize}

  Adding the bounds in the above three cases, we have
  \begin{equation*}
    \eps_k \le n^{-\nu} k^{-\nu k}.
  \end{equation*}
  Then the lemma follows.
  \qedhere
\end{proof}

Then we use the above bound to prove \cref{eq:v-k-v-k'-sum}.
\begin{lemma}
  \label{lem:v-k-v-k'-sum}
  With parameters satisfying \cref{main-parameters}, w.p. $1 - o(1)$,
  \begin{equation*}
    \abs{ \sum_{k = 0}^\cut V_k z^k - \bigsum_{k = 0}^\cut V_k' z^k}
    = \mathcal{O}(n^{c-\nu}).
  \end{equation*}
\end{lemma}

\begin{proof}
  Let $M = \frac{\ln n}{\ln \ln n}$.
  Using \cref{lem:eps-bound}, w.p. $1 - o(1)$,
  \begin{equation}
    \begin{split}
      \abs{\sum_{k=0}^\cut \eps_k z^k} \le &
      \sum_{k=0}^M \eps_k \abs{z}^k + \sum_{k=M+1}^\cut \eps_k \abs{z}^k\\
      \le & n^{-\nu} \sum_{k=0}^M \abs{z}^k + n^{-\nu} \sum_{k=M+1}^\cut
      k^{-\nu k} \abs{z}^k.
    \end{split}
  \end{equation}

  For large $n$,
  \begin{equation*}
    \sum_{k=0}^M \abs{z}^k \le 2 (\ln n)^{cM} \le 2 n^c.
  \end{equation*}

  On the other hand, since $M = \frac{\ln n}{\ln\ln n}$ and
  $c < \nu$ as in \cref{main-parameters},
  it holds for large $n$ that
  \begin{equation}
    \nonumber
    \frac{(k+1)^{-\nu (k+1)} \abs{z}^{k+1}}{k^{-\nu k} \abs{z}^{k}} =
    (k+1)^{-\nu} \abs{z} \left( 1 + \frac{1}{k} \right)^{-\nu k} \le
    \frac{\abs{z}}{(M+1)^\nu} < \frac{1}{2}.
  \end{equation}
  Thus, for large $n$,
  \begin{equation}
    \sum_{k = M + 1}^t k^{-\nu k} \abs{z}^k \le M^{-\nu M} \abs{z}^M \le n^c,
  \end{equation}
  which means
  \begin{equation}
    \bigsum_{k = 0}^\cut \eps_k \abs{z}^k = \mathcal{O}(n^{c-\nu}).
  \end{equation}
\end{proof}

\bibliographystyle{alpha}
\bibliography{refs}

\end{document}